\documentclass[11pt,letterpaper]{article}
\usepackage{amsmath,amsthm}
\usepackage[margin=1in]{geometry}
\usepackage[dvipsnames]{xcolor}
\usepackage{hyperref}
\usepackage{authblk}

\title{Certified-Everlasting Quantum NIZK Proofs}
\date{}

\author{Nikhil Pappu \\
\href{mailto:nikpappu@pdx.edu}{\color{magenta}{\small{nikpappu@pdx.edu}}}}
\affil{Portland State University}

\usepackage{mathtools}
\usepackage{comment}
\usepackage{braket}
\usepackage{framed}
\usepackage{amssymb}
\usepackage{graphicx}
\usepackage{color, xcolor}
\usepackage{hyperref}
\usepackage{cleveref}

\definecolor{darkyellow}{RGB}{184, 135, 11}

\hypersetup{
    colorlinks=true,
    citecolor=magenta,
    linkcolor=MidnightBlue,
    urlcolor=cyan
}

\newtheorem{theorem}{Theorem}[section]

\theoremstyle{definition}
\newtheorem{definition}[theorem]{Definition}

\theoremstyle{plain}

\newtheorem{lemma}[theorem]{Lemma}

\newtheorem{corollary}[theorem]{Corollary}

\newtheorem{remark}[theorem]{Remark}

\makeatletter
\renewenvironment{proof}[1][\proofname]{%
  \par\begingroup
  \pushQED{\qed}%
  \normalfont
  \topsep6\p@\@plus6\p@\relax
  \addvspace{\topsep}%
  \noindent\textit{#1\@addpunct{.}}\enspace\ignorespaces
}{%
  \popQED\par\addvspace{\topsep}%
  \endgroup
}
\makeatother

\newcommand{\cL}{\mathcal{L}}
\newcommand{\NP}{\mathsf{NP}}
\newcommand{\QMA}{\mathsf{QMA}}
\newcommand{\secp}{\lambda}
\newcommand{\crs}{\mathsf{crs}}
\newcommand{\wit}{\omega}

\newcommand{\trho}{\widetilde{\rho}}
\newcommand{\qP}{\rho_{\Pro}}
\newcommand{\qV}{\rho_{\Ver}}

\newcommand{\tqV}{\widetilde{\rho}_{\Ver}}
\newcommand{\RL}{R_\cL}
\newcommand{\cert}{\mathsf{cert}}

\newcommand{\negl}{\mathsf{negl}(\secp)}

\newcommand{\poly}{\mathsf{poly}}
\newcommand{\la}{\gets}
\newcommand{\ra}{\to}
\newcommand{\tx}{\widetilde{x}}

\newcommand{\tpsi}{\widetilde{\psi}}
\newcommand{\tphi}{\widetilde{\phi}}
\newcommand{\tsigma}{\widetilde{\sigma}}
\newcommand{\tpi}{\widetilde{\pi}}
\newcommand{\seteq}{\coloneq}
\newcommand{\bit}{\{0,1\}}
\providecommand{\bits}{\bit}
\newcommand{\Exp}{\mathsf{Exp}}
\newcommand{\tExp}{\widetilde{\mathsf{Exp}}}
\newcommand{\sPi}{\mathsf{\Pi}}

\newcommand{\tsPi}{\widetilde{\sPi}}
\newcommand{\com}{\mathsf{com}}
\newcommand{\open}{\mathsf{open}}
\newcommand{\nizk}{\mathsf{nizk}}

\newcommand{\sfs}{\mathsf{s}}

\newcommand{\D}{\mathsf{D}}
\newcommand{\td}{\mathsf{td}}

\newcommand{\GenBits}{\mathsf{GenBits}}
\newcommand{\Verify}{\mathsf{Ver}}
\newcommand{\cC}{\mathcal{COM}(\secp)}
\newcommand{\Open}{\mathsf{Open}}
\newcommand{\HBG}{\mathsf{HBG}}
\newcommand{\bg}{\mathsf{bg}}
\newcommand{\hb}{\mathsf{hb}}
\newcommand{\op}{\mathsf{op}}
\newcommand{\cezk}{\mathsf{cezk}}
\newcommand{\zk}{\mathsf{zk}}

\newcommand{\qreg}[1]{{\ensuremath{\textcolor{gray}{\mathrm{#1}}}}}

\newcommand{\inn}{\mathsf{in}}
\newcommand{\out}{\mathsf{out}}

\newcommand{\ct}{\mathsf{ct}}
\newcommand{\msg}{\mathsf{m}}
\newcommand{\TD}{\mathsf{TD}}
\newcommand{\pad}{\mathsf{pad}}
\newcommand{\tpad}{\mathsf{\widetilde{\pad}}}

\newcommand{\Sone}{\mathsf{S}_1}
\newcommand{\Stwo}{\mathsf{S}_2}
\newcommand{\sS}{\mathsf{S}}

\newcommand{\sz}{\mathsf{z}}
\newcommand{\aux}{\mathsf{aux}}

\newcommand{\tsz}{\widetilde{\mathsf{z}}}

\newcommand{\Test}{\mathsf{Test}}
\newcommand{\chk}{\mathsf{Check}}
\newcommand{\Rout}{\mathsf{R}_\mathsf{out}}
\newcommand{\sft}{\mathsf{t}}

\newcommand{\F}{\mathsf{F}}
\newcommand{\key}{\mathsf{k}}

\newcommand{\qA}{\mathcal{A}}
\newcommand{\qD}{\mathcal{D}}
\newcommand{\qR}{\mathcal{R}}
\newcommand{\qB}{\mathcal{B}}

\newcommand{\sD}{\mathsf{D}}
\newcommand{\rhoA}{\rho_{\qreg{A}}}
\newcommand{\rhoB}{\rho_{\qreg{B}}}

\newcommand{\rhoAB}{\rho_{\qreg{A}\otimes\qreg{B}}}
\newcommand{\cZ}{\mathcal{Z}}
\newcommand{\cX}{\mathcal{X}}
\newcommand{\cY}{\mathcal{Y}}
\newcommand{\cO}{\mathcal{O}}

\newcommand{\Hyb}{\mathsf{Hyb}}
\newcommand{\ttheta}{\widetilde{\theta}}

\newcommand{\Setup}{\mathsf{Setup}}
\newcommand{\Pro}{\mathsf{P}}
\newcommand{\tPro}{\widetilde{\mathsf{P}}}
\newcommand{\Ver}{\mathsf{V}}
\newcommand{\tVer}{\widetilde{\mathsf{V}}}
\newcommand{\Del}{\mathsf{Del}}

\newcommand{\DelVrfy}{\mathsf{Cert}}
\newcommand{\Cert}{\mathsf{Cert}}
\newcommand{\Sim}{\mathsf{Sim}}
\newcommand{\prfk}{\mathsf{prfk}}
\newcommand{\cH}{\mathcal{H}}
\newcommand{\sig}{\mathsf{sig}}
\newcommand{\tsig}{\widetilde{\mathsf{sig}}}

\newcommand{\Tr}{\mathsf{Tr}}
\newcommand{\norm}[1]{\left\lVert#1\right\rVert}
\newcommand{\ketbra}[2]{\mathinner{|{#1}\rangle\,\langle{#2}|}}

\begin{document}

\maketitle

\begin{abstract}
We study non-interactive zero-knowledge proofs (NIZKs) for $\NP$
satisfying: 1) statistical soundness, 2) computational zero-knowledge (ZK) and 3) certified-everlasting zero-knowledge (CE-ZK).
The CE-ZK property allows a verifier of a quantum proof to revoke
the proof in a way that can be checked (certified) by the prover. Conditioned on
successful certification, the verifier's state can be
efficiently simulated with only the statement, in a \emph{statistically
indistinguishable} way. Our contributions regarding these
certified-everlasting NIZKs (CE-NIZKs) are as follows:

\begin{itemize}
\item We identify a barrier to obtaining CE-NIZKs in the CRS model via
generalizations of known interactive ZK proofs that satisfy CE-ZK.

\item We circumvent this by constructing CE-NIZK 
using non-black-box use of NIZK for $\NP$ satisfying certain
properties, along with OWFs. As a result, we obtain
CE-NIZKs for $\NP$ in the CRS model, based on polynomial hardness of the learning with errors
(LWE) assumption.

\item In addition, we observe that the aforementioned barrier does not
apply to the shared EPR model. We leverage this to construct
CE-NIZK for $\NP$ in this model based on any statistical binding hidden-bits
generator, which is known from LWE. The only quantum computation
here involves single-qubit measurements of the EPR pairs.
\end{itemize}
\end{abstract}

\newpage

{\hypersetup{linkbordercolor=black}
\hypersetup{linkcolor=black}
\tableofcontents
}

\newpage


\section{Introduction}\label{sec:intro}

A non-interactive zero-knowledge proof (NIZK) \cite{STOC:BluFelMic88} allows a
\emph{prover} to demonstrate the validity of a statement to a
\emph{verifier} with a \emph{single} message, without revealing
any other information. As an example, if Alice knows a solution to a
Sudoku puzzle, she can convince Bob that the puzzle is indeed
solvable, without leaking the solution itself. Crucially, if the
puzzle does not actually have a solution, then a cheating Alice should not be
able to fool Bob, a guarantee formally known as \emph{soundness}.  On
the other hand, the inability of a cheating Bob to learn any additional
information is referred to as \emph{zero-knowledge}. This latter
property is formalized by the existence of an efficient algorithm
called a \emph{simulator}, that can recreate Bob's view using the statement
alone. Intuitively, it captures that whatever Bob learns
from a proof, he can compute efficiently without it.

NIZKs have been a cornerstone of cryptography with widespread
theoretical and practical impact. For example,
they play an important role in constructions of CCA-secure PKE 
\cite{STOC:NaoYun90}, digital signatures \cite{C:BelGol89},
ZAPs \cite{FOCS:DwoNao00},
round-efficient MPC \cite{EC:AJLTVW12},
and functional encryption \cite{FOCS:GGHRSW13} to name a few. On the
practical front, they have been instrumental in designing anonymous credentials
\cite{belenkiy2009randomizable}, group signatures \cite{ACNS:BCCGG16},
cryptocurrencies \cite{sasson2014zerocash} and verifiable computation
\cite{C:GenGenPar10}.

\paragraph{Setup Models.} It is well-known that NIZKs suffer a
drawback compared to their multi-round (interactive) counterparts in
that they are impossible to obtain in the plain model
\cite{JC:GolOre94}. Fortunately, this can be circumvented in a realistic
setup model called the common reference string (CRS) model. Here, a
trusted party samples a random (possibly non-uniform) public string
before the protocol, and security is required to hold with overwhelming probability
over the choice of CRS. Other commonly studied setup models include the designated verifier
\cite{EC:QuaRotWic19} and random oracle models \cite{CCS:BelRog93}.
In the quantum world, recent works \cite{AC:MorYam22,EC:ABKK23} have
considered two-party protocols where the parties share different
registers of an entangled quantum state. These are sometimes known to
provide stronger guarantees than the CRS model \cite{C:BarKhuSri23}.
Recently, more complex yet localized quantum setups have also
garnered attention in two-party protocols \cite{TCC:AnaGulLin24,C:MorNehYam24}.

\paragraph{Unbounded Adversaries.} Due to their widespread use cases,
NIZKs satisfying various security properties have been explored. Of
these, perhaps the most important characterization is in regards to
whether security holds against an unbounded adversary, or only a
computationally-bounded one. For instance, \emph{soundness} typically
only requires that an \emph{efficient} malicious prover cannot
convince the verifier of a \emph{false} statement. Likewise, the
standard notion of \emph{zero-knowledge} requires that an
\emph{efficient}
verifier learns nothing
from the transcript. However, one can consider upgrades of these in the form of \emph{statistical soundness} and
\emph{statistical zero-knowledge}, which provide security against an
unbounded prover and an unbounded distinguisher (equivalently, the
verifier in the NIZK case) respectively. Hence, many NIZKs in the
literature satisfy these stronger properties
\cite{STOC:BluFelMic88,FLS90,FLS99,JC:GolOre94,C:GroOstSah06,STOC:SahWat14,TCC:BitPan15,STOC:CCHLRR19,C:PeiShi19,EC:WatWeeWu25}.

\paragraph{Statistical Soundness or Zero-Knowledge?.} Although
statistical variants of soundness and zero-knowledge are achievable
from standard assumptions, it is unlikely to obtain both
simultaneously, except for a
small subclass of $\NP$ languages \cite{PS05}. This is the case
even for interactive-protocols \cite{STOC:Fortnow87}, with evidence
that quantum protocols do not help either \cite{MW18}. Consequently, one
has to choose which of these two properties to give up in favor of the
other, depending on the application. Typically, statistical zero-knowledge
is preferred as soundness only needs to hold during the online phase, while zero-knowledge can be compromised anytime in the
future to leak sensitive information. Still, this is not always
ideal. Consider a large corporation (prover) that might have various means to
cut down the time to break certain computational assumptions, with the
verifier being an average user. In this case, a computational bounded
assumption on the verifier could be more reasonable. A statistically
sound NIZK can also remove the need to refresh the CRS
periodically in such a setting. However, this introduces
other issues, such as the verifier being able to sell the data
to a powerful rival corporation.

\paragraph{A Quantum Compromise.} Such a
tradeoff seems inevitable in the classical setting, but previous works
have proposed a novel compromise using quantum resources, in the
context of interactive zero-knowledge \cite{C:HMNY22,C:BarKhu23}. In
particular, the prover and verifier engage in a statistically-sound
quantum protocol, after which the verifier is asked to return a
\emph{certificate of deletion}. Then, an everlasting
guarantee akin to statistical zero-knowledge holds, given that the
certificate is verified to be valid. Note that this only makes sense
in a quantum world due to the no-cloning theorem, as the verifier can
retain a classical proof indefinitely.

Observe also that the requirement to
delete the proof does not interfere with verification. This is because
verification can be performed first without disturbing the state (by
gentle measurement), and the proof can be deleted having served its
purpose. Another interesting aspect is that the verifier can
immediately verify and revoke the proof, which can then be checked.
As a result, the parties need not maintain the
state for long unlike other primitives with
certified-deletion. This can be quite useful due to the difficulty of
maintaining quantum states.

The aforementioned works then showed
that such proofs can be obtained for all languages in $\QMA$ (the quantum
generalization of $\NP$) assuming only OWFs, or even weaker
quantum-assumptions such as the existence of pseudo-random states
\cite{C:JiLiuSon18}.
Naturally, one can
consider the non-interactive analogue where the prover sends a single
message to the verifier, and the verifier sends back a deletion
certificate. While one might argue this is no better than a two-round
protocol, the deployed applications wouldn't be slowed down by the
second message, which could be batched and sent at a later point.
Hence, at the expense of using quantum resources and relying on the
ability to penalize a non-confirming verifier, there exists the
possibility of obtaining statistically-sound NIZKs with an everlasting
zero-knowledge guarantee. Despite the hope, and
the construction of many advanced primitives with
certified-everlasting forms of security (See Section
\ref{sec:rel-work}), CE-NIZKs are yet to be realized. We explore this
possibility in detail.

\subsection{Contributions}

\begin{enumerate}
\item \emph{Definitions and Impossibility.} We present
definitions of certified-everlasting NIZK proofs (CE-NIZKs) in the CRS
and shared-EPR models to begin with. Then, we examine natural
approaches to constructing CE-NIZKs in the CRS model. Specifically, we
identify a class of protocols we call \emph{deletion-resistant}
CE-NIZKs. At a high level, these admit a deletion algorithm that
splits the proof state into two registers, one that passes
verification of deletion, and another that passes proof verification.
We then rule out the existence of
such protocols by showing that they imply NIZKs that are both
statistically-sound and statistically zero-knowledge. We also argue
that non-interactive counterparts of known interactive
protocols fall into this category.

\item \emph{Feasibility in the CRS Model.} As a result of this barrier, we 
provide a new approach to constructing CE-NIZKs in the CRS model,
where we employ two ``layers'' of NIZK proofs that are generated in
quantum-superposition. In order to prove security, we employ the
two-slot technique \cite{STOC:NaoYun90} and the OR proof strategy
\cite{FLS90,FLS99}, and reduce to a certified-deletion
theorem of BB84 states \cite{C:BarKhu23}. This allows us to construct a
CE-NIZK for $\NP$ based on (non-black-box use of) a superposition-secure NIZK
proof for $\NP $\footnote{We
require two special notions from such a NIZK: superposition-secure adaptive
zero-knowledge and post-quantum non-adaptive zero-knowledge with statistical CRS
indistinguishability which are satisfied by the 
LWE-based NIZKs of \cite{STOC:Waters24,EC:WatWeeWu25}.}
(Definition \ref{def:nizk}), along with just post-quantum OWFs.
We observe that
superposition-secure NIZKs (a stronger requirement than post-quantum
NIZKs) are implied by
LWE due to the construction of \cite{EC:WatWeeWu25}. Hence, we
obtain our result based on LWE, matching
the state-of-the-art assumption for post-quantum NIZK.

\item \emph{Bypassing the Impossibility with Shared Entanglement.} We
then explore a natural quantum setup where the prover and
verifier share entanglement in the form of halves of EPR pairs.
Surprisingly, we find that the impossibility regarding
\emph{deletion-resistant} protocols no longer
applies to this setting. As a consequence, we show that a
generalization of the hidden-bits compiler of \cite{EC:QuaRotWic19}
suffices to obtain CE-NIZKs for $\NP$ in this setting. In particular, we obtain
CE-NIZKs based on any statistically binding hidden-bits generator
\cite{EC:QuaRotWic19}, which can be based on LWE
\cite{EC:WatWeeWu25}. Although the assumption is the same as our CRS
protocol, the protocol is highly efficient in regards to the quantum
operations involved. For comparison, our CRS protocol generates and
verifies proofs in superposition of a highly entangled state. On the
other hand, this protocol only requires single-qubit measurements of
the EPR halves in one of two basis, and uses classical computation and
communication otherwise.

\end{enumerate}

\subsection{Related Work}\label{sec:rel-work}

\paragraph{Certified-Everlasting Cryptography.}

The pioneering work of Broadbent and Islam \cite{TCC:BroIsl20} constructed
encryption with certified-deletion, where ciphertexts can be deleted
in a way that the plaintext becomes unrecoverable (by an unbounded
adversary) even if the secret-key is leaked later. This was followed
by a series of works on certified-everlasting security, which refers
to a bounded adversary that cannot break some 
guarantee in the future using unbounded computation, assuming it provides a valid
certificate during the computationally bounded stage. 

The work of Hiroka, Morimae, Nishimaki and Yamakawa \cite{C:HMNY22}
constructed a certified-everlasting statistical-binding commitment
scheme and used it to obtain a certified-everlasting zero-knowledge proof
system for $\QMA$. Later, the work of Bartusek and Khurana
\cite{C:BarKhu23} introduced an
elegant compiler for obtaining primitives with certified-deletion.
Through it, they obtained PKE, ABE, FHE, commitments (and ZK via
\cite{C:HMNY22}), witness encryption, timed-release encryption, MPC
protocols and secret-sharing schemes with certified-everlasting
security. The work of Hiroka et al. \cite{EC:HKMNPY24} constructed public-key
functional encryption (and other variants), compute-and-compare
obfuscation and garbled circuits with such security guarantees.
Advanced secret-sharing schemes with certified-everlasting security
were studied in the work of Bartusek and Raizes \cite{C:BarRai24}.
The work of Champion,
Kitagawa, Nishimaki and Yamakawa \cite{CKNY25} constructed a certified-everlasting
variant of untelegraphable encryption, a relaxation of unclonable
encryption. A notion of differing inputs obfuscation (dIO) with
certified-everlasting security was introduced and constructed in the
work of Bartusek et al. \cite{EC:BGKMRR24}. Their construction
provides the publicly-verifiable deletion property \cite{TCC:KitNisYam23,BKP23,BKMPW23,KS25}.

\paragraph{Quantum-Enabled NIZKs.} In recent years, several works have
studied NIZKs with interesting properties that are only feasible in a
quantum world. Some of these works construct
schemes that prevent copying attacks, i.e., they ensure that the
recipient of a proof cannot convince several other verifiers, without
itself knowing a witness. These ``unclonable'' NIZKs and variants were
studied in the concurrent works of Jawale and Khurana
\cite{AC:JawKhu24} and
Goyal, Malavolta and Raizes \cite{TCC:GoyMalRai24}. A related notion
of NIZK with certified-deletion (NIZK-CD) was
explored by the work of Abbaszadeh and Katz \cite{AC:AbbKat25} where the
recipient must generate a valid proof of deletion, and cannot convince
another verifier thereafter. While the unclonable variants were shown
to imply public-key quantum money, NIZK-CD was constructed in
\cite{AC:AbbKat25} from standard assumptions.
In Section \ref{sec:to}, we informally argue that CE-NIZK implies a form of
NIZK-CD (without the extractability notion considered in \cite{AC:AbbKat25}). We
also argue NIZK-CD does not imply CE-NIZK.

The work of Çakan, Goyal and Raizes
\cite{CGR26}
studied NIZK with a stronger form of certified-deletion called
certified-deniability. Intuitively, this is a simulation-based notion
that guarantees that the state after deletion of the proof could be
computed without the proof. The work presents a feasibility for
this in the random oracle model (ROM), with evidence against its
existence in the plain model. Our simulation definition shares
some similarity to theirs, while additionally requiring statistical
indistinguishability. Crucially, their impossibility does not
contradict our CRS feasibility, as our simulator has the
ability to manipulate the CRS. In the certified-deniability context,
this ability would be restricted, similar to the work's 
restriction on programmability in the ROM.

Finally, the work of \cite{STOC:GLRRV25} studies a notion of one-time
security for NIZK and presents a construction from iO and LWE. This
notions allows an authority to distribute quantum proving tokens to a
prover. Then, the prover can only convince a verifier of as many
statements as the number of tokens received. This notion is shown to
imply public-key quantum money which justifies the use of the strong
assumption of iO. We believe the guarantees of all these NIZK notions
(apart from NIZK-CD) to be orthogonal to that of CE-ZK. For a detailed
comparison of our work with some of these works, see Section
\ref{sec:to_other}.

\paragraph{Revocable Cryptography.} Revocable cryptography is closely
related to the notions discussed above in that it involves a
deletion (revocation) phase, after which certain
information/capability is lost. This was pioneered by the work of
Unruh \cite{EC:Unruh14} 
which introduced revocable time-released encryption. Revocable notions
for encryption/signatures were also studied in \cite{AMP25,MPY24}, and are
similar to the notions of certified-deletion in the sense that
secret data is being deleted. On the other hand, a long line of works
on secure software leasing \cite{EC:AnaLaP21,TCC:KitNisYam21,TCC:BJLPS21,CMP24} and
secure key leasing
\cite{AC:KitNis22,TCC:AnaPorVai23,EC:AKNYY23,AC:PWYZ24,TCC:AnaHuHua24,EC:CGJL25,EC:KitMorYam25,C:KitNisPap25,KNP26,KLYY26}
consider the following high-level premise. A user is provided a
quantum state that can be used to evaluate some software or
cryptographic functionality. Later, the user can provide a
deletion-certificate, after which it is guaranteed that they can no
longer evaluate the functionality. We remark that
almost all these schemes only consider a computationally-bounded
second-stage (post-deletion) adversary. This restriction is necessary in
several contexts like the leasing of decryption keys of a PKE scheme.


\section{Technical Overview}\label{sec:to}

\subsection{Defining Certified-Everlasting NIZK}

We will begin by defining the syntax of a certified-everlasting
non-interactive zero-knowledge proof (CE-NIZK) (Definition
\ref{def:ce-nizk}). A CE-NIZK for a
language $\cL \in \NP$ is a tuple of four algorithms $(\Setup, \Pro,
\Ver, \DelVrfy)$ described as follows. The classical algorithm
$\Setup(1^\secp)$ outputs a common reference string $\crs$ as usual.
The quantum prover's algorithm $\Pro$ on input the CRS $\crs$, a statement $x \in
\cL$ and corresponding witness $\wit \in \RL(x)$, outputs a quantum
proof state $\sigma$, along with a (possibly classical) quantum
state $\qP$, which will later be used to validate returned
proofs. The quantum verifier's algorithm $\Ver$ on input $\crs$, $x$ and proof
$\sigma$ outputs a decision bit $b$ along with a quantum state $\qV$.
The quantum certification algorithm $\DelVrfy$
takes as input quantum states $\qP, \qV$ and outputs $\top$ (accept)
or $\bot$ (reject). The context is that the verifier would check the
validity of the proof $\sigma$ to obtain $\qV$ (these states would be
close to one another in trace distance due to completeness and gentle-measurement) and send it to the prover who checks
that the proof is revoked by running $\DelVrfy(\qP, \qV)$. In general, 
$\DelVrfy$ may end up destroying $\qP$ and $\qV$. Note that when we refer to a CE-NIZK for $\NP$, we mean one that is a CE-NIZK
for an NP-complete language.

We require a CE-NIZK for $\NP$ to satisfy completeness,
statistical soundness and computational zero-knowledge, which must also hold
for a standard NIZK proof for $\NP$. In addition, it must also
satisfy certified-everlasting zero-knowledge (CE-ZK).
Intuitively, the notion guarantees that the view of any QPT malicious
verifier $\Ver^*$ can be simulated in a statistically close manner
with just the statement $x$, conditioned on
certification being successful. Formally, for every QPT verifier
$\Ver^*$, there must exist a QPT simulator $\Sim_{\Ver^*}$ such that the
following guarantee holds for every true statement $x$ and
corresponding witness $\wit$:

\[
\left[
\trho \approx_s \rho'
 \ :
\begin{array}{rl}
 &\crs \la \Setup(1^\secp)\\
 & (\sigma, \qP) \la \Pro(\crs, x, \wit)\\
 & \rho_{\qreg{A}\otimes\qreg{B}} \la \Ver^*(\crs, x, \sigma)\\
 & \textrm{If}\;\DelVrfy(\qP,\rho_{\qreg{A}}) = \top \;
 \textrm{then}\; \trho
 \seteq \rho_{\qreg{B}}.\\
 & \textrm{Else}\; \trho \seteq \bot.\\
 & \rho' \la \Sim_{\Ver^*}(x)
\end{array}
\right]
\]

Let us explain this formalism. First, a $\crs$ is sampled using
$\Setup$ and a real proof $\sigma$ and prover state $\qP$ are
computed. Next, the proof $\sigma$ is provided to the malicious
verifier, who outputs a bi-partite state $\rhoAB$. Then, the notion
sets a state $\trho$ to be $\rhoB$ (the residual state on register
$\qreg{B}$ after $\DelVrfy$ is applied to register $\qreg{A}$) only if the state $\rhoA$ on register $\qreg{A}$ passes
certification by $\DelVrfy(\qP, \rhoA)$. If not, $\trho$ is simply set to
$\bot$, thereby giving a distinguisher no additional information when
the certification is unsuccessful. Then, we require that the
simulator $\Sim_{\Ver^*}$ on input $x$ produces a simulated state
$\rho'$ that is close in trace distance to $\trho$ from the real
execution. Note that our protocols will come with a simulator that makes
black-box use of $\Ver^*$, but we define a non-black-box one here for
generality. 

\subsection{Non-Existence of Deletion-Resistant Protocols}\label{sec:del_res_imp}

Before we introduce our approach to realizing CE-NIZKs, we will
discuss some challenges in achieving them.
Specifically, we identify a natural class of CE-NIZKs that
are impossible to construct. The intuition is that these come with a
deletion algorithm that splits the proof state into two, the first which passes
prover certification, and a second which passes $\NP$ verification.
Consequently, this operation can be performed at
the prover's end before sending the second part to the verifier.
We observe
that these imply NIZKs which are both statistically sound and statistically
zero-knowledge which are unlikely to exist. We define this subclass of CE-NIZKs
formally as follows.

\begin{definition}[Deletion-Resistant CE-NIZKs] \label{def:dr-ce-nizk}A non-interactive
protocol with the syntax of CE-NIZK (Definition \ref{def:ce-nizk}) is
said to be a deletion-resistant CE-NIZK if it satisfies
the CE-ZK property, along with the following properties. There exist
QPT algorithms $\Ver^*, \tVer$ such that
the following conditions hold for every $x \in \cL$, $\wit \in \RL(x)$ and
$\tx \notin \cL$, where $(\crs, \sigma, \rho_\Pro, \rhoAB)$ are
sampled as $\crs \la \Setup(1^\secp), (\sigma, \rho_\Pro) \la
\Pro(\crs, x, \wit)$ and $\rhoAB \la \Ver^*(\sigma)$:
\begin{flalign}
&\Pr\left[
\DelVrfy(\qP, \rhoA) \ra
\top
\right] \ge 1/\poly(\secp)\\
&\Pr\left[
\tVer \big(\crs, x, \rhoB) \ra 1 \; \Big\lvert \;
\DelVrfy(\qP, \rhoA) \ra
\top
\right] \ge 1 - \negl\\
&\Pr\left[
    \exists \tsigma : \widetilde{\Ver}(\crs, \widetilde{x},
    \tsigma) \ra 1
\right] \le \negl
\end{flalign}

\end{definition}

\begin{theorem}\label{thm:del_res_imp}
If there exists a deletion-resistant CE-NIZK for a language $\cL$,
there is a quantum NIZK for $\cL$ with
statistical soundness and statistical ZK. 
\end{theorem}
The proof is simple, so we show it in Section \ref{sec:crs}.
Also, even if (2) holds with
only inverse polynomial probability, it results in a NIZK proof but 
with
inverse polynomial notions of completeness and statistical ZK. However, we believe an extension
of the argument of Pass and Shelat \cite{PS05} works even in this case
(See the third paragraph in Section \ref{sec:to_other}),
which would show that the language $\cL$ belongs to $\mathsf{QAM} \cap
\mathsf{coQAM}$ (and hence unlikely to be $\NP$-complete).

%
%
\subsection{On Generalizations of Prior Work}\label{sec:gen}
We now recall previous works on CE-ZK proofs in the interactive
setting by Hiroka et al. \cite{C:HMNY22} and
Bartusek and Khurana \cite{C:BarKhu23}. The work of Hiroka et al. constructed
CE-ZK proofs for languages in $\QMA$ based on a new primitive they
introduced called commitments with
statistical binding and certified-everlasting hiding. The
certified-everlasting hiding property refers to the fact that after
deletion, even an unbounded receiver is unable to learn anything about
the contents of the commitment. While such commitments were
constructed in the random oracle model in their work, Bartusek and
Khurana constructed them from any statistical binding commitment
scheme in the plain model. Note that such statistically binding
commitments can be based on OWFs \cite{Nao91} and from potentially
weaker quantum-cryptographic assumptions \cite{C:MorYam22,C:AnaQiaYue22}.

We now argue that these protocols do not generalize easily for
constructing CE-NIZKs. At a high level, the CE-ZK proofs established
by these works rely on the following template:
\begin{itemize}
\item First, the prover sends some $n = \poly(\secp, |x|,
|w|)$ commitments $\{\ket{\com^i}\}_{i\in[n]}$ to the
verifier, generated using the certified-everlasting commitment scheme.
\item Then, the verifier sends a randomly chosen opening set $S \subset
[n]$ to the prover, along with deletion certificates
$\{\cert^i\}_{i\in[n]\setminus S}$ corresponding to the commitments
that are not to be opened.
\item Finally, the prover sends the openings of the
commitments $\{\open^i\}_{i\in S}$ for the set $S$. These openings are
then sufficient for the verifier to determine the validity of the
statement.
\end{itemize}
Note that while these works use classical certificates
$\{\cert^i\}_{i\in[n]\setminus S}$, the same ideas also apply if
these certificates were quantum. The main idea here is that since
the unopened commitments are deleted by the verifier, even an
unbounded distinguisher cannot tell apart the real and simulated
views. Naturally, one might expect that a
Fiat-Shamir based approach \cite{C:FiaSha86,STOC:CCHLRR19,C:PeiShi19} in the CRS model could be employed to turn these interactive
protocols into non-interactive ones.
However, we argue that such protocols would be deletion-resistant, and
hence run into Theorem \ref{thm:del_res_imp}. To see this, consider a
CE-NIZK $\sPi =
(\Setup, \Pro, \Ver, \DelVrfy)$ with a similar template, where the
prover $\Pro$ determines the opening set $S$ by itself, by applying a
hash function to some classical part of the generated proof.
Since $\Ver$
would only need to check the commitments in $S$ as part of the
protocol, we can consider $\tVer \seteq \Ver$ and a malicious QPT
verifier $\Ver^*$ that splits the commitments in $S$ and the
commitments in $[n] \setminus S$ into two different registers.
Consequently, the state $\rhoAB$ produced by
$\Ver^*$ is such that $\rhoA$ would pass certification while $\rhoB$
would be accepted by $\Ver$. Notice that both verification and
certification could check the validity of $S$, but the classical
information that is hashed to produce $S$ can simply be cloned.
Furthermore, designing $\Ver, \DelVrfy$ to
artificially deviate from the template doesn't help. For e.g.,
$\Ver$ cannot check all the commitments, as $\sPi$
should satisfy computational zero-knowledge. Moreover, even if
$\DelVrfy$ requires all the commitments to be revoked, $\Ver^*$ can
clone the contents of the opened commitments.
We note that natural attempts at generalizing the hidden-bits
paradigm \cite{FLS90,FLS99} to the CE-NIZK setting do not work, due to a
similar argument. 

\subsection{Feasibility in the CRS Model}

Despite the aforementioned barrier, we show that CE-NIZKs exist in the
CRS model based on standard assumptions. To this end, our first idea
is to utilize a superposition of NIZK proofs entangled with some
unclonable quantum state. The hope is that this at least prevents the
splitting attacks mentioned in the previous subsections. We explain
the structure of this approach next.

\paragraph{Two-Layers of NIZK Proofs.} Consider a BB84 state
$\ket{y}^\theta = H^{\theta[1]}\ket{y[1]} \otimes \ldots
\otimes H^{\theta[\secp]}\ket{y[\secp]}$, where $y,\theta \la
\bit^\secp$ and $H$ is the Hadamard transform. In the work of Bartusek
and Khurana \cite{C:BarKhu23}, it was shown
that such states satisfy a certified deletion property. Specifically,
if an adversary $\qA$ is given $\ket{y}^\theta$, and $\qA$
produces a deletion certificate $\cert \in \bit^\secp$ such that
$\cert[i] = y[i]$ for all
$i:\theta[i] = 1$ (i.e., the values of the Hadamard basis positions
are consistent), then the computational basis parity $p =
\bigoplus_{i:\theta[i] = 0} y[i]$
remains statistically close to uniform in the view of $\qA$, even if the ``secret-key''
$\theta$ is revealed to $\qA$ after $\cert$ is produced.
In our construction, a long BB84 state of $\ell \cdot \secp$ qubits is
used,
i.e., $\ket{y}^\theta \seteq \ket{y^1}^{\theta^1} \otimes \ldots
\otimes \ket{y^\ell}^{\theta^\ell}$, where each $y^j, \theta^j \la
\bit^\secp$. In the scheme, the prover $\Pro$ uses a classical NIZK for $\NP$ 
$\sPi_\nizk \seteq (\Setup_\nizk, \Pro_\nizk,
\Ver_\nizk)$ and first
generates an ``inner'' NIZK proof $\pi_\inn \la \Pro_\nizk(\crs_\inn, x,
\wit)$ wrt a CRS $\crs_\inn \la \Setup_\nizk(1^\secp)$. Then, $\pi_\inn$ is masked as $\ct \seteq
\pi_\inn \oplus \pad \oplus \key$, where $\key$ is a uniform string\footnote{This
uniform value is used due to a technicality in reducing to the theorem
of \cite{C:BarKhu23}.} and $\pad \seteq p_1 \| \ldots \| p_\ell$,
where each $p_i$ is the computational-basis parity bit wrt the $i$-th
BB84 state. The intuition is that if $\ket{y}^\theta$ is ``deleted'',
then $\pi_\inn$ cannot be recovered. However, the state
$\ket{y}^\theta$ and ``ciphertext'' $\ct$ are not useful by
themselves, as $\ct$ now just appears uniform. Consequently, we
employ a second ``outer layer'' of proofs entangled with
$\ket{y}^\theta = \sum_\sz \alpha_\sz \ket{\sz}$. Specifically, we
generate $\pi^\sz_\out \la \Pro_\nizk(\crs_\out, x^\sz, \wit_\out \seteq
(\theta, \key))$ for each superposition term $\sz$ wrt the same CRS
$\crs_\out \la \Setup(1^\secp)$, for the following ``meta-statement'' $x^\sz$:
$$\underline{\text{Meta-Statement}\; x^\sz}: \;\;\exists \big(\theta,
\key\big) \; : \; \Ver_\nizk\Big(\crs_\inn, x, \ct
\oplus f(\theta, \sz) \oplus \key\Big) = 1$$
with $f(\theta,\sz) \seteq \bigoplus_{i: \theta^1[i]=0} \sz^1[i] \|
\ldots \| \bigoplus_{i: \theta^\ell[i]=0} \sz^\ell[i]$ where $\theta,
\sz \in \bit^{\ell\secp}$ and $\theta =
\theta^1 \| \ldots \| \theta^\ell$ and $\sz = \sz^1 \| \ldots \|
\sz^\ell$. Observe now that $f(\theta, \sz) = \pad$ for each $\sz$ in
the superposition of $\ket{y}^\theta = \sum_\sz\alpha_\sz \ket{\sz}$.
Now, using the proof state $\sigma \seteq \big(\sum_\sz\alpha_\sz
\ket{\sz}\ket{\pi^\sz_\out}, \ct = \pi_\inn \oplus \pad \oplus
\key \big)$, the verifier can check in superposition that the proofs
$\pi_\out^\sz$ are valid wrt the statements $x^\sz$, by executing
$\Ver_\nizk(\crs_\out, x^\sz, \pi^\sz_\out)$. Completeness follows
easily, because $\ct \oplus f(\theta, \sz) \oplus \key = \pi_\inn$
which is accepted by $\Ver_\nizk(\crs_\inn, x, \pi_\inn)$ by
completeness of $\sPi_\nizk$. Next, consider the case when $x$ is
false, which means there is no accepting proof $\tpi_\inn$ that makes
$\Ver_\nizk(\crs_\inn, x, \cdot)$ accept (whp over choice of
$\crs_\inn$). As a result, all of the meta-statements $x^\sz$ would
also be false and hence would be rejected whp by
$\Ver_\nizk(\crs_\out, x^\sz, \cdot)$. Note that due to adaptive
soundness of $\sPi_\nizk$, it is not a problem even if $x^\sz$ are
computed adversarially based on $\crs_\out$. As for computational ZK,
it follows since $(\crs_\inn, \pi_\inn)$ can be replaced with
$(\crs_\Sim, \pi_\Sim)$ which is guaranteed by the computational ZK of
$\sPi_\nizk$. We also remark that since the scheme proves statements
involving the description of $\Ver_\nizk$, it makes non-black-box use
of $\sPi_\nizk$. Hereafter, we will focus on achieving the
certified-everlasting zero-knowledge (CE-ZK) property.

\paragraph{Reduction to Certified-Deletion Theorem.}
In order to enforce CE-ZK, we require the verifier to return (revoke) the
entire proof state $\sigma = (\sum_\sz \alpha_\sz \ket{\sz}_{\qreg{R}}
\ket{\pi_\out^\sz}_{\qreg{P}}, \ct)$. Then, the certification algorithm $\Cert$
uncomputes the proofs $\pi^\sz_\out$ in superposition, followed by
measuring register $\qreg{R}$ in the Hadamard basis to get a classical
outcome $\cert$. Finally, $\Cert$ checks if $\cert$ is consistent with
$y$ at the Hadamard basis positions specified by $\theta$.
Now, our goal is to prove CE-ZK by reducing to the
aforementioned certified-deletion theorem of BB84 states
\cite{C:BarKhu23}. This theorem is more general than previously discussed 
in that a QPT adversary $\qA$ can be given an arbitrary state $\ket{\psi}$ that is
generated by some distribution $\cZ(\theta, b \oplus p, \ket{y}^\theta)$ where 
$y, \theta \la \bit^\secp$, $p \seteq \bigoplus_{i:\theta[i]=0}y[i]$
and $b \in \bit$. Then, conditioned on QPT $\qA$ outputting a valid certificate $\cert$
(consistent with the Hadamard positions), the advantage in
distinguishing $b = 0/1$ using $\qA$'s leftover state is
negligible, even for unbounded distinguishers. However, this guarantee
only holds for distributions $\cZ$ which computationally hide the basis information
$\theta$, i.e., the pre-condition $\cZ(\theta, b', \ket{\phi}) \approx_c
\cZ(0^\secp, b', \ket{\phi})$ must hold for all $\theta, b',
\ket{\phi}$. Next, we describe how this distribution $\cZ$ is defined in our context:
\begin{description}
\item $\underline{\cZ(\theta, b', \ket{\phi}):}$
\begin{itemize}
\item The state $\ket{\phi}$ is parsed as $\ket{\phi} = \sum_\sz
\alpha_\sz \ket{\sz}_{\qreg{R}}$ on register $\qreg{R}$.
\item Then, the proof state $\sigma \seteq
(\sum_{\sz}\alpha_\sz \ket{\sz}_{\qreg{R}}\ket{\pi_\out^\sz}_{\qreg{P}}, \ct)$ is
computed using \textcolor{black}{$(\theta, b')$} (mostly similar to
the construction, but we expand on it shortly).
\item Next, the QPT $\Ver^*$ (attacking CE-ZK) is run on $\sigma$,
producing a state $\rho_{\qreg{A\otimes B}}$ on registers $\qreg{A},
\qreg{B}$, with $\qreg{A}$ denoting the ``revocation'' register.
\item Let $\rho_{\qreg{A}}$ denote the state on register $\qreg{A}$ at
this point, which is parsed as a state on registers $\qreg{R \otimes
P}$ (recall that honest revocation returns the proof).
\item Finally, the proofs $\pi_\out^\sz$ on register $\qreg{P}$ are
uncomputed, resulting in the final state $\rho'_{\qreg{R \otimes P \otimes
B}}$ on the three registers. This state $\rho'$ is then output.
\end{itemize}
\end{description}
The context here is as follows. Consider a QPT malicious verifier
$\Ver^*$ (which defines $\cZ$) and an unbounded distinguisher $\D$
breaking CE-ZK. Then, in the BB84 deletion game wrt $\cZ$, the
reduction $\qR$ receives the state $\rho'_{\qreg{R\otimes P \otimes
B}} \la \cZ(\theta, b \oplus p, \ket{y}^\theta)$, after which it
simply measures register $\qreg{R}$ in the Hadamard basis to produce
$\cert$ and outputs the state on register $\qreg{B}$. If $\Ver^*$
passes certification in the CE-ZK game, then $\cert$ will be a valid
certificate in the BB84 game. We will discuss
shortly that this allows to statistically switch a bit in the CE-ZK
game, as otherwise, $\D$ would distinguish $b=0/1$ in the
BB84 game using the output of $\qR$. Before that, we discuss how the pre-condition is satisfied.

The key idea we use to establish the pre-condition is that even though $\cZ$
depends on $\theta$ as the witness $\wit_\out$ for generating the proofs
$\pi^\sz_\out$, the computational ZK of $\sPi_\nizk$ can be used to
simulate these proofs without access to $\theta$. Note however, that there are
exponentially many proofs in superposition, generated wrt the
same CRS $\crs_\out$. Still, we can rely on the superposition-secure
adaptive ZK property (Definition \ref{def:nizk}) for this simulation.
Consequently, we define $\cZ$ (See Lemma \ref{lemma:main}) so that it
generates simulated proofs in superposition (instead of the real
proofs $\pi_\out^\sz$), when its first input is $0^\secp$. Note that
this simulation can also handle the uncomputation step above, allowing to
show $\cZ(\theta, b', \ket{\phi}) \approx_c \cZ(0^\secp, b',
\ket{\phi})$. This in-turn allows to switch a bit $b=1$ to $b=0$ in the
CE-ZK security game, in a \emph{statistically indistinguishable} way
(conditioned on successful certification). Next, we elaborate on what this bit flip actually does in our context.
\paragraph{Utilizing Two-Slots.}
The purpose of the statistically indistinguishable bit flip is to
``zero-out'' the inner proof $\pi_\inn$ (i.e., replacing it
with $0^\secp$) in a bit-by-bit manner, thereby removing the dependence on the witness
$\wit$. The problem with this bit flip is that even though the proofs
$\pi_\out^\sz$ satisfy the NP relation wrt the statements
$x^\sz_\out$ and witness $\wit_\out = (\theta, \key)$ when considering
the distribution $\cZ(\theta, b'=1, \ket{\phi})$, this no longer holds
for $\cZ(\theta, b'=0, \ket{\phi})$. This is because the unmasked inner
proof $\pi_\inn'$ in the latter case may no longer make
$\Ver_\nizk(\crs_\inn, x, \pi_\inn')$ accept. For this reason, we use a
$2\ell\secp$-qubit BB84 state instead and modify the meta-statement with a dummy slot, described as follows:

\begin{align*}
\underline{\text{Meta-Statement}\; x^\sz}: \;\;\exists \big(\theta,
\key^0, \key^1\big) \; : \; \Ver_\nizk\Big(\crs_\inn, x, \ct^0&
\oplus f^0(\theta, \sz) \oplus \key^0\Big) = 1\\
&\text{OR}\\
\Ver_\nizk\Big(\crs_\inn, x, \ct^1&
\oplus f^1(\theta, \sz) \oplus \key^1\Big) = 1
\end{align*}
where $f^0$ is defined similarly to $f$ for the last $\ell$-qubits
(which encode $\pad^0$) and $f^1$ corresponds to the first $\ell$-qubits (which encode
$\pad^1$), with $\ct^0 \seteq \pi_\inn \oplus \pad^0 \oplus \key^0$
and $\ct^1 \seteq 0^\secp \oplus \pad^1 \oplus \key^1$.

Note that the reason for performing a check using $\Ver_\nizk$ in the
dummy slot is because we cannot simply prove a trivial statement in
the simulation. For example, we cannot have the real CRS be a
commitment to $1$ and use a commitment to $0$ in the simulation. This
is because the simulated CRS distribution must be statistically close
to the real CRS distribution (because the classical CRS cannot be
``deleted''), on top of the proof being statistically sound.
Consequently, we will rely on the non-adaptive computational ZK with
statistical CRS indistinguishability property (See Definition
\ref{def:nizk})
of $\sPi_\nizk$. Intuitively, this means that the marginal distribution
of the simulated CRS $\crs_\Sim$ is statistically close to that of the
real CRS $\crs_\inn$, while the overall ZK guarantee is still
computational. We observe in Section \ref{sec:nizk} that the LWE-based
NIZK proofs of \cite{STOC:Waters24,EC:WatWeeWu25}
obtained via the hidden-bits approach satisfy this property. For
current LWE-based NIZK proofs obtained via
the Fiat-Shamir approach \cite{STOC:CCHLRR19,C:PeiShi19}, this is not
satisfied as the CRS includes a
``fake'' hash key that is only computationally close to a real key.

Using this property, we first switch $\crs_\inn$ into $\crs_\Sim$
statistically.
Then, due to the nature of the OR meta-statement, we are able to swap the $0^\ell$ plaintext
in $\ct^1$ into a simulated proof $\pi_\Sim$ (generated along with
$\crs_\Sim$) one bit at a time. This
is because we no longer have the aforementioned issue as $\pi_\inn$ inside $\ct^0$ is
sufficient to satisfy the OR condition, and hence invoke
superposition-secure
zero-knowledge of the outer proofs. Then, once $\pi_\Sim$ is placed in
$\ct^1$, we can replace $\pi_\inn$ bit-by-bit into $0^\ell$ since
$\pi_\Sim$ also satisfies $\Ver_\nizk(\crs_\Sim, x, \pi_\Sim) = 1$ by
the computational zero-knowledge of $\sPi_\nizk$. We remark that this is the only place
where computational ZK \emph{for the inner layer proof} is used in the proof of CE-ZK. Finally, observe that the resulting distribution can be
produced with the statement $x$ alone (without $\wit$), ensuring that the CE-ZK
property is satisfied.

\paragraph{Some Subtleties.} We remark that instead of the inner proof
$\pi_\inn$, we could encrypt the witness $\wit$ directly in the
construction and check the NP relation in the first slot of the
meta-statement. However, we chose the current variant as we anyway use
the $\Ver_\nizk$ check for the second slot and the proof of
ZK is simpler in this case.

Additionally, apart from a NIZK for $\NP$, we rely on a quantum-secure PRF (implied by OWFs)
to simulate uniform randomness in superposition, used to compute and
uncompute the outer NIZK proofs.

We also face the following issue in
the reduction to the deletion theorem. The quantum
state returned by the adversary may also have superposition terms
$\widetilde{\sz}$ such that $\wit_\out = (\theta, \key^0, \key^1)$ is not a valid
witness for statement $x^{\widetilde{\sz}}$. Consequently, we cannot
argue that the state after uncomputation of real proofs is
computationally close to the state after uncomputation of simulated
proofs, based on the superposition-secure ZK property.
To circumvent this, we employ a simpler (information-theoretic)
variant of the key-testability technique of
\cite{TCC:KitNisYam23,C:KitNisPap25}. At a high level, this involves
producing an information-theoretic MAC tag for each term of the BB84
state, which can be verified after revocation. Intuitively, this
ensures that the adversary cannot produce a state with new
superposition terms (that have non-negligible amplitude) without
breaking unforgeability of the MAC. In summary, we obtain the
following theorem:

\begin{theorem}[CE-NIZK in the CRS Model]
There exist CE-NIZKs for $\NP$ in the CRS model
satisfying 1) statistical soundness, 2)
computational zero-knowledge, and 3) certified-everlasting zero-knowledge,
assuming the existence of OWFs and NIZKs for $\NP$ in the CRS model satisfying 1) 
statistical (adaptive) soundness, 2) computational (non-adaptive)
zero-knowledge with statistical CRS indistinguishability, and 3)
superposition-secure (adaptive) zero-knowledge.
\end{theorem}

Since such NIZKs are implied by LWE (See Section \ref{sec:nizk}) due to
\cite{EC:WatWeeWu25}, we have:

\begin{corollary}[CE-NIZK from LWE]
There exist CE-NIZKs for $\NP$ in the CRS model based on the
polynomial hardness of the LWE problem.
\end{corollary}

\subsection{Quantum-Efficient Protocol in the Shared EPR Model}

Previously, we argued that generalizations of the Fiat-Shamir and
hidden-bits paradigms to the CE-NIZK setting do not work, as they run
into Theorem \ref{thm:del_res_imp}. We observe that this impossibility does not apply to the shared EPR
model where the parties are initialized with halves of
polynomially-many EPR pairs. Such a model has recently garnered
attention in the context of non-interactive MPC \cite{EC:ABKK23,C:BarKhuSri23} and NIZK for
QMA \cite{C:BarKhuSri23,AC:MorYam22}. Recall that Theorem
\ref{thm:del_res_imp} relies on the following idea: a prover could
split the quantum proof the verifier is supposed to receive into two
parts, one which passes certification while the other suffices for
verification. Hence, the prover need only send the part which is
required for verification. This then implies a statistically secure
NIZK. In the EPR model, the prover cannot perform
this ``split'' as it does not control the entire quantum proof state, since
the verifier's EPR halves do not originate from the prover.

Our observation has some similarities to the work of
\cite{C:BarKhuSri23}, which separated the CRS and shared EPR models in
the context of non-interactive randomized oblivious transfer.  In our
case, the separation between the models only applies to the
aforementioned deletion-resistant protocols, as evident from our CRS
feasibility result. However, we demonstrate that this difference is
useful as well, by constructing a protocol in the shared EPR model that is
much more efficient in its use of quantum resources. Specifically, our
protocol only involves performing single-qubit computational or Hadamard basis
measurements on the EPR halves, and classical computation
otherwise. We proceed to describe our approach.

We make use of the well-known hidden-bits paradigm \cite{FLS90,FLS99} that we adapt
to the shared EPR model for the context of obtaining
certified-everlasting ZK. We note that Fiat-Shamir variants of
previous interactive protocols would also bypass the impossibility in
the EPR model. However, we find the hidden-bits approach to be cleaner
and also helps to avoid the statistically far simulated CRS
issue with Fiat-Shamir approaches mentioned earlier. We also
emphasize that we assume a common reference string in addition to the
shared EPRs. This is because it is not clear how to simulate a
structured reference string using EPRs unlike a common random string.
Moreover, statistically sound NIZKs from LWE
\cite{STOC:CCHLRR19,C:PeiShi19,STOC:Waters24,EC:WatWeeWu25} are all in the
common reference string model. While we could consider a more complex
entangled state that can simulated arbitrary reference strings, we
stick with EPR setups for their simplicity and practicality. We
believe that the NIZK proof of Bitansky and Paneth \cite{TCC:BitPan15}
based on iO and OWFs would also work for our purpose, allowing to use
EPRs alone as it only requires a uniform CRS.

The hidden bits paradigm for constructing NIZK utilizes two
components: (1) a hidden bit generator (HBG) and (2) a statistically
sound and statistically zero-knowledge NIZK in the hidden-bits model.
In the (idealized) hidden-bits model, there is a uniformly random string that is
sampled and made available to the prover, but is hidden from the
verifier. Then, the prover has the ability to specify some indices
of its choice to a trusted third party, which then reveals the corresponding
hidden bits to the verifier (and nothing about the other bits). Since
NIZKs in the hidden-bits model are known unconditionally \cite{FLS90,FLS99},
most works focus on designing an HBG. Intuitively, the purpose of the
HBG is to realize the idealized hidden bits in the CRS model. In
particular, it allows the prover $\Pro$ to sample a succinct
commitment $\com$, which along with the CRS, fixes a hidden bit
string. The HBG then allows the prover to provably open certain
positions of this fixed hidden bit string to the verifier.
Crucially, the \emph{statistical-binding} security of the HBG ensures
that the prover is bound to this string, and cannot reveal anything
else to the verifier. Moreover, \emph{computational hiding}
ensures that all unopened hidden bits are pseudorandom in the
verifier's view.

Our idea is that instead of determining the actual hidden bits string via
the commitment $\com$ and the CRS $\crs$, these are utilized
to fix a common basis string $\theta$ for both parties. Then, the
prover would reveal parts of $\theta$ to the verifier, and both
parties would obtain the actual hidden bits by measuring their halves
of the EPR pairs in the $\theta$ basis. Specifically, the EPRs are
divided into polynomially many ``blocks''. Then, the $i$-th block is used to derive the
$i$-th hidden bit, set to be the XOR of computational basis
measurements of the positions $j$ where $\theta[j]=0$. Note that the nature of EPR
pairs ensures that both parties would derive the same measurement
outcomes on measuring in the same basis. Also, statistical soundness is
ensured due to the observation that we can consider an equivalent
verifier that measures all its EPRs in the computational basis before the
prover is even initialized. Consequently, the prover can only
manipulate the hidden bits via the basis string $\theta$ which is
restricted by the HBG commitment $\com$. Note that the inability of
this hypothetical verifier to pass certification is not a problem, as
it is only used as a proof tool to argue soundness. While
computational zero-knowledge follows easily by arguments similar to
prior work \cite{EC:QuaRotWic19}, certified-everlasting zero-knowledge is argued
based on the following fact. The verifier is asked to delete all the
unopened hidden bits by measuring the corresponding EPRs in the Hadamard basis.
Consequently, even if computational hiding of the HBG eventually leaks
$\theta$ to an unbounded distinguisher, the
certified-deletion theorem of \cite{C:BarKhu23} can be applied to
ensure that these hidden bits are information-theoretically hidden.
This is because shared EPR pairs exhibit properties similar to
BB84 states. Consequently, we have the following theorem:

\begin{theorem}
CE-NIZKs for $\NP$ exist in the shared EPR Model (Definition
\ref{def:ce_nizk_epr}), assuming statistical binding hidden bit
generators (Definition \ref{def:hbg}) exist.
\end{theorem}

From the fact that statistical binding HBGs are known from LWE
\cite{EC:WatWeeWu25}, and due to the structure of our protocol, we have the
following corollary:

\begin{corollary}
Assuming the polynomial hardness of LWE, there exist CE-NIZKs for
$\NP$ in the shared EPR Model, where the prover and verifier only
perform single-qubit measurements apart from arbitrary classical
computation.
\end{corollary}

\subsection{Relation with Other Quantum-Enabled NIZKs}\label{sec:to_other}

\paragraph{Unclonable NIZKs.}
The works of \cite{AC:JawKhu24} and \cite{TCC:GoyMalRai24} introduced
unclonable NIZK which is a primitive that at a high level, guarantees
the following property. Consider an adversary that is provided with a
quantum proof of some $\NP$ statement $x$. Then, if it can produce two
separate quantum proofs that both pass verification,  one can
efficiently extract out a witness $w$ for the statement $x$ from the
adversary.  Intuitively, this ensures that if an adversary only has
the proof and not the witness, it cannot produce more than one valid
proof. Clearly, this is only feasible in the quantum setting due to
the no-cloning principle.
In \cite{AC:JawKhu24}, it was also informally mentioned that such
proofs are a
prerequisite for obtaining NIZKs with some form of certified deletion,
while also mentioning certified everlasting security. The intuition is
that if the adversary could clone proof states, then it could delete one of
them successfully while retaining the other. Moreover, these works show that
unclonable NIZK implies public-key quantum money (PKQM). On the other hand,
we establish CE-NIZK in this paper from LWE, which is not known to
imply PKQM. The reasoning is that CE-NIZKs need not be unclonable in
the sense of \cite{AC:JawKhu24}, which actually doesn't require the adversary
to generate two copies of the state. To see this, observe the
following attack on our CE-NIZK in the CRS model that clearly breaks
unclonable security. The adversary on obtaining the proof state
$\sigma = \sum\limits_{\sft}\alpha_{\sft} \ket{\sft}$,
simply applies CNOTs to get the state
$\sum\limits_{\sft}\alpha_{\sft}\ket{\sft}\ket{\sft}$ and gives the two entangled
registers to two verifiers. Since verification is done entirely in the
computational basis, both verifiers would accept the proof. Natural
``unentangled'' attacks also work, which involving collapsing
$\sigma$ and sending the outcome to both verifiers. We also
mention that by utilizing coset states along with iO (which implies
PKQM assuming OWFs) instead of BB84
states, one might be able to obtain both unclonability and CE-ZK. Note that
unclonability also doesn't imply CE-ZK in general, because the
unclonability adversary never receives unbounded computational power.

\paragraph{NIZKs with Certified-Deletion.}

In the related work of \cite{AC:AbbKat25}, a NIZK with
certified-deletion (NIZK-CD) was obtained from standard assumptions,
specifically from a classical NIZK (with certain properties), and a
statistically binding commitment. The security intuitively ensures
that if an adversary is able to pass the deletion check and simultaneously
produce an accepting proof, then one can efficiently extract a witness
from the adversary. We will discuss a natural weakening of this
definition, where the adversary should simply be unable to
output a valid certificate and a valid proof, for statement-witness
pairs drawn from some hard distribution. We refer to this weaker
variant as wNIZK-CD. We consider this variant for comparison because
CE-NIZK is not defined with an extractability guarantee in
our work. We now argue informally CE-NIZK implies wNIZK-CD.

\begin{remark}In all our discussions regarding wNIZK-CD and NIZK-CD, we assume the
schemes are statistically sound, as it is the focus of our work.\end{remark}

\paragraph{CE-NIZK $\Rightarrow$ wNIZK-CD.}
Consider some language $L \in \NP$ and some sufficiently hard distribution $D$ over statement-witness
pairs. Let $\sPi \seteq (\Setup, \Pro, \Ver, \Cert)$ be a
CE-NIZK scheme for $L$. Consider a QPT $\qA$ that is given $\sigma \la
\Pro(\crs, x, \wit)$ for $(x, \wit) \la D$, after which it produces
$\rho_{\qreg{A}\otimes\qreg{B}}$ such that $\Cert(\rho_\Pro,
\rho_{\qreg{A}}) = 1$ and $\Ver(\crs, x, \rho_{\qreg{B}}) = 1$. Assume
for now that this condition holds with probability $1$. We can now
construct a QPT prover $\Pro^*$ that on input $(x, \wit) \la D$,
computes $\sigma \la \Pro(\crs, x, \wit)$ and
$\rhoAB \la \qA(\sigma)$, followed by applying $\Cert(\rho_\Pro, \rhoA)$
and finally outputting $\rhoB$ as the proof. Observe that the NIZK
$(\Setup, \Pro^*, \Ver)$ satisfies completeness (on average wrt $D$),
statistical soundness (a requirement of CE-NIZK), and statistical ZK
(on average, and from the CE-ZK of $\sPi$). We believe that such a
proof system leads to contradiction. To see why, recall that
Pass and Shelat \cite{PS05} show that if a language $L$ has a
NIZK with both statistical soundness and statistical ZK, then $L \in
\mathsf{AM} \cap \mathsf{coAM}$. Even if the event
$\Cert(\rho_\Pro, \rhoA) = 1 \land \Ver(\crs, x, \rhoB) = 1$ occurs
with only inverse polynomial probability, thereby implying inverse
polynomial variants of completeness and statistical ZK,
the proof of \cite{PS05} goes through. This is because the
completeness of their $\mathsf{coAM}$ protocol follows
from $\negl$-soundness of the NIZK, while inverse polynomial soundness
follows from the completeness and statistical ZK of the NIZK. We note
that a similar result would also hold for analogous (average case)
quantum complexity classes, which are unlikely to contain sufficiently
hard $(L, D)$.

Next, we argue the reverse direction does not hold.

\paragraph{wNIZK-CD $\nRightarrow$ CE-NIZK.}

Consider a QPT verifier $\Ver^*$ and unbounded distinguisher $D$ that
break the CE-ZK property of $\sPi$. The problem is that $D$ is
unbounded, but a wNIZK-CD adversary must efficiently generate an
accepting proof. Another issue is that $D$ need not recover the entire
witness. Hence, it is unclear how a reduction would work.
We believe explicit counterexamples can also be shown, by
considering statistically sound extensions of existing NIZK-CD schemes
\cite{AC:AbbKat25}. In other words, these NIZK-CD schemes would not satisfy
CE-ZK. This is because the classical NIZK proofs attached
to the quantum proof state in these schemes can be compromised by unbounded
attackers.

\subsection{Open Questions}

\begin{enumerate}
\item \emph{Publicly-Verifiable Deletion:} While our construction only
satisfies privately-verifiable deletion, obtaining the publicly-verifiable
variant is more desirable. We believe this should be possible by
employing coset states \cite{C:CLLZ21} instead of BB84 states, along
with iO (by employing techniques from \cite{EC:BGKMRR24}).
Note that coset states with iO and OWFs are sufficient to obtain
public-key quantum money (PKQM). However, it is not clear that CE-NIZK with
publicly-verifiable deletion implies PKQM because an efficient PKQM
verifier may be unable to exploit the fact that an adversarial state contains
some information about the witness. Still, it is not obvious how to
obtain publicly-verifiable deletion without iO, and we leave it as an
interesting open question.

\item \emph{Classical Certificates}: Our construction in the CRS model
requires quantum revocation, while it is ideal if the verifier could
only send back a classical certificate. One approach
towards achieving this would be to utilize a hybrid quantum fully
homomorphic encryption (QFHE) scheme that is known from LWE
\cite{Mah20,C:Brakerski18}, as used in the works of
\cite{STOC:Shmueli22,C:Shmueli22}. In particular, the verifier could
be provided with a QFHE ciphertext that encrypts $(\theta, \key^0,
\key^1)$ so that it could perform the uncomputation of the
superposition proofs using QFHE evaluation. This way, it would be able
to obtain a certificate that is a QFHE encryption of Hadamard basis
measurements, which could be decrypted and checked
by the prover. The problem here is that we cannot simply replace the
QFHE plaintext with a dummy, followed by using a similar
certified-deletion argument as \cite{C:BarKhu23}. This is because
deletion verification
requires the QFHE secret-key, making it unclear how to reduce to the
semantic security of QFHE. Despite this, there might be new techniques that would allow to prove the
security of such a construction, and we leave this for future work.

\item \emph{Certified-Everlasting NIWIs in the Plain Model}: It is
also interesting to consider other non-interactive proof systems such
as NIWIs, which are possible to obtain even in the plain model
\cite{FOCS:DwoNao00,C:BarOngVad03}. In the post-quantum setting,
these can be based on post-quantum iO \cite{TCC:BitPan15}. In
particular, one can try to obtain a NIWI with statistical soundness along with
computational witness indistinguishability (WI) and
certified-everlasting WI in the plain model.
Our construction seems to be general enough to
work for this case, except that the outer proofs
generated in superposition can no longer be simulated by the computational
zero-knowledge property. Recall that this is an important part of our proof, as it
is needed to invoke the certified deletion theorem. It might still be
possible just to rely on the witness indistinguishability property.
However, we think this approach would require a stronger deletion
theorem, and possibly techniques from leakage-resilient cryptography
\cite{KR19,TCC:CGLR24}.

\item \emph{CE-NIZKs for QMA}: In this
work, we focus on obtaining CE-NIZK for $\NP$, due to the challenges
in obtaining NIZK for $\QMA$ in the CRS model (See
\cite{C:Shmueli21}). On the other hand, interactive CE-ZK is known for
all of $\QMA$ \cite{C:HMNY22}. In this regard, it would be interesting
to explore CE-NIZK for $\QMA$ in other setup models. For inspiration, the
works of \cite{C:Shmueli21} and \cite{AC:MorYam22} show NIZKs for
$\QMA$ in the designated-verifier and shared-EPR models respectively, by
utilizing variants of the Fiat-Shamir paradigm. On the other hand, one
could consider CRS model constructions that make heuristic use of iO \cite{BJRT26}.
\end{enumerate}


\section{Preliminaries}\label{sec:prelim}

\subsection{Notation}

We denote a (possibly quantum) algorithm $\Pro$ with input $1^\secp$
and outputs $\msg, \ct$ as $\Pro(1^\secp) \ra (\msg, \ct)$ or $(\msg,
\ct) \la \Pro(1^\secp)$. Here, $\secp$ denotes a security parameter
and $1^\secp \seteq 1 \| \ldots (\secp \text{ times}) \| 1$, where
$\|$ denotes concatenation and $\seteq$ denotes defining (also used
for assigning/substituting). We use $\Pro(1^\secp)[1]$ and
$\Pro(1^\secp)[2]$ to denote the first and second indices of the
output tuple respectively. For a set $S$, by $s \la S$, we mean $s$ to
be sampled uniformly from $S$. For a distribution $D$, by $s \la D$,
we mean sampling $s$ according to $D$.
The functions $\poly(\secp)$ and $\negl$
denote a polynomial and a negligible function in $\secp$
respectively. We use the notation $\qA^{\ket{\Pro(\crs,
\cdot, \cdot)}}$ to denote that the quantum algorithm $\qA$ gets
quantum access to a quantum algorithm $\Pro(\crs, \cdot, \cdot)$ (with
classical inputs and outputs) on
the second and third inputs. For a language $\cL \in \NP$, we consider
the relation $\RL$ consisting of statement-witness pairs. Hence,
$\RL(x)$ denotes the set of witnesses of a statement $x$. We denote
quantum registers as $\qreg{A}, \qreg{B}, \qreg{R}$ etc.
We use the notation $y_i$ to denote the $i$-th bit of a string $y$ (unless
it is explicitly defined to be something else), and
$\qreg{R}_i$ to denote the $i$-th qubit of $\qreg{R}$. Sometimes, we
also use $y[i]$ to denote the $i$-th bit. Additionally,
we use the notation $y_I$ to denote $\{y_i\}_{i\in I}$ for some set
$I$. Note that $y^i$ doesn't mean the $i$-th bit of $y$. However, the
notation $y^I$ denotes $\{y^i\}_{i \in I}$. Sometimes, we use $\overline{I}$ to denote the complement
of a set $I$, when the universe is clear from context. We also employ
the notation $s_I \seteq t_I \oplus r_I$ to denote $s_i \seteq t_i
\oplus r_i$ for each $i \in I$.
For a classical algorithm $\mathsf{C}$, we use
$\mathsf{C}(s; r)$ to denote running it on input $s$ and random tape
$r$. We denote pure quantum states as $\ket{\psi}, \ket{\sz}$ etc
while the quantum states $\rho, \sigma$ will either refer to pure or
mixed states based on the context. We use the notation
$\ket{\psi}_{\qreg{R}}$ to denote a state on register $\qreg{R}$ and $\rhoAB$ to
denote a bi-partite state on register $\qreg{A}\otimes\qreg{B}$. We
denote the state on $\qreg{B}$ obtained by tracing out $\qreg{A}$ as
$\rhoB \seteq \Tr_{\qreg{A}}(\rhoAB)$.
For simplicity, we often use $\rhoB$ to denote the state
on register $\qreg{B}$ at different points when it is clear from
context. For e.g., we do not use a different $\rhoB'$ to refer to
the state obtained after $\qreg{A}$ is
collapsed to some specific value.
An EPR pair is
the following state on a pair of registers $\qreg{P}, \qreg{V}$:
$\frac{1}{\sqrt2}\big(\ket{0}_{\qreg{P}}\ket{0}_{\qreg{V}} +
\ket{1}_{\qreg{P}}\ket{1}_{\qreg{V}} 
\big)$, where $\qreg{P}, \qreg{V}$ are often referred to as EPR
halves.
For two distributions $D_0,
D_1$, by the notation $D_0 \approx_c D_1$, we mean the two are
computationally indistinguishable by quantum polynomial time
(QPT) adversaries with non-uniform quantum advice. Often, we
simply use $d_0 \approx_c d_1$ for samples $d^0 \la D_0, d^1 \la D_1$
drawn from these distributions to denote the same. We use $D_0 \equiv
D_1$ to denote the distributions are identical and $D_0 \approx_s D_1$
to denote they are statistically close. Likewise, for quantum states
$\rho$ and $\sigma$, $\rho \approx_s \sigma$ denotes that the states are
negligibly close in trace distance.

\subsection{Quantum Information}

A pure state $\ket{\psi}$ is a vector in a complex Hilbert space
$\mathcal{H}$ with unit norm $\norm{\ket{\psi}} = 1$. An
operator $P$ which satisfies $P^\dagger = P$ is said to be a Hermitian
operator. Let $S(\mathcal{H})$ be the set of Hermitian operators on
$\mathcal{H}$. A density matrix $\rho \in S(\mathcal{H})$ is a
positive semi-definite operator with unit trace $\Tr(\rho) = 1$. A
density matrix represents a mixture over pure states, and the density
matrix of a pure state $\ket{\psi}$ is $\ketbra{\psi}{\psi}$. A
Hilbert space can be split into registers $\cH \seteq \cH_{\qreg{R_1}}
\otimes \cH_{\qreg{R_2}}$. A unitary matrix $U$ is one that satisfies
$UU^\dagger = U^\dagger U = \mathbb{I}$. Applying $U$ to a state
$\ket{\psi}$ transforms it into the state $U\ket{\psi}$. For a mixed
state $\rho$, applying $U$ transforms it into $U \rho U^\dagger$. A
projector $\mathsf{\Pi}$ is a Hermitian operator additionally
satisfying $\mathsf{\Pi}^2 = \mathsf{\Pi}$. The trace distance between
two mixed states $\rho, \sigma$ is given by
$\TD\big(\rho, \sigma) \seteq \frac12 \Tr(\sqrt{(\rho -
 \sigma)^{\dagger}(\rho - \sigma)}\big)$. It provides an
upper bound on the distinguishing advantage of $\rho, \sigma$ by any
quantum algorithm. A quantum polynomial-time (QPT) algorithm refers to
a polynomial-size quantum circuit composed of some universal gate set,
along with some non-uniform advice state $\ket{\psi}$. 
We denote the
Hadamard transform/operator by $H$, which is defined as
$
H = \frac{1}{\sqrt{2}} \begin{pmatrix}
1 & 1 \\
1 & -1
\end{pmatrix}
$.
For $y, \theta
\in \bit^n$ for some polynomial $n$ in $\secp$, we consider BB84
(Wiesner) states of the following form, where $y_i, \theta_i$ denote
the $i$-th bits of $y, \theta$ respectively:

$$\ket{y}^\theta \seteq H^{\theta_1}\ket{y_1} \otimes \ldots
\otimes H^{\theta_n}\ket{y_n}$$

\subsection{Cryptography Basics}

\begin{definition}[Learning with Errors (LWE) \cite{STOC:Regev05}]
Consider integer functions $n,m,q \in \mathbb{Z}$ of the security
parameter $\secp$. Let $\chi(\secp)$ be an error
distribution over $\mathbb{Z}$. The LWE assumption is said to hold if
for every QPT adversary $\qA$, there exists a negligible function
$n(\secp)$ such that:

\[
\Bigg\lvert
\Pr\left[
\qA\big(A, y\big) \ra 1
\ \bigg\lvert
\begin{array}{rl}
&A \la \mathbb{Z}^{n\times m}_q\\
&s \la \mathbb{Z}^n_q\\
&e \la \chi^m\\
&y \seteq s^\intercal A + e
\end{array}
\right]
-
\Pr\left[
\qA\big(A, u\big) \ra 1
\ \bigg\lvert
\begin{array}{rl}
&A \la \mathbb{Z}^{n\times m}_q\\
& u \la \mathbb{Z}^m_q\\
\;
\end{array}
\right]
\Bigg\rvert
\le n(\secp)
\]
\end{definition}

\begin{definition}[Quantum-Secure Pseudo-Random Function]\label{def:qprf}
Let $\{\F_\prfk : \bit^{\ell_\inn} \ra
\bit^{\ell_\out}\}_{\prfk\in \bit^\secp}$
be an efficiently computable family of functions, where $\ell_\inn,
\ell_\out$ are some polynomials in $\secp$. Then,
$\{\F_\prfk\}_\prfk$ is said to be a quantum-secure PRF if for every QPT adversary $\qA$,
there exists a negligible function $n(\secp)$ such that the following
holds, where $\mathcal{R}$ is
the set of all functions with input $\ell_\inn$ and output
$\ell_\out$:

\[
\bigg\lvert
\Pr\left[
    1 \la \qA^{\ket{f(\cdot)}}(1^\secp) \;\big\lvert\; f \la \mathcal{R}
\right] 
-
\Pr\left[
1 \la \qA^{\ket{\F_\prfk(\cdot)}}(1^\secp) \; \big\lvert \;
    \prfk \la \bit^{\secp}
\right]
\bigg\rvert
\le n(\secp)
\]
\end{definition}

\begin{theorem}[\cite{FOCS:Zhandry12}] Quantum-secure PRFs exist, assuming
post-quantum secure one-way functions exist.
\end{theorem}

\begin{lemma}[One-Way to Hiding (O2H) \cite{C:AmbHamUnr19}]\label{lma:o2h}
Consider any functions $G, H: \cX \ra \cY$, an arbitrary value $z$,
and some set $S \subseteq \cX$ such that for every $x \notin S$, $G(x) =
H(x)$. The values $G, H, S, z$ may be arbitrarily distributed.
Consider now an oracle algorithm $\qA$ that makes at most $q$ queries
to its oracle. Let $\qB$ be an oracle algorithm such that $\qB^H$ executes as
follows on input $z$:

$\underline{\qB^H(z):}$

\begin{enumerate}
\item Choose $i \la [q]$.
\item Execute $\qA^H(z)$ until it makes its $i$-th query.
\item Measure the $i$-th query and output the outcome.
\end{enumerate}

Then, the following holds:

$$\bigg\lvert \Pr\Big[\qA^H(z) = 1\Big] - \Pr\Big[\qA^G(z) =
1\Big]\bigg\rvert \le 2q \cdot \sqrt{\Pr\Big[\qB^H(z) \in S\Big]}$$

\begin{remark}
It is easy to see from the proof of \cite{C:AmbHamUnr19} that the lemma holds even
when the auxiliary input $z$ is some quantum state $\rho_z$. A similar
observation was made in \cite{AC:HhaXagYam19}, where it was mentioned
that $z$ could even be exponentially large, and hence allows to describe
quantum states as well.
\end{remark}

\end{lemma}

\subsection{Superposition-Secure Non-Interactive
Zero-Knowledge}\label{sec:nizk}

\begin{definition}\label{def:nizk}
A triple of PPT algorithms $(\Setup, \Pro, \Ver)$ is
a superposition-secure non-interactive zero-knowledge proof (NIZK) in the common
reference string (CRS) model for a language $\cL$, if
it satisfies the following properties:

\begin{itemize}
\item \emph{Completeness:} The following holds for every $x \in \cL$ and
$\wit \in \RL(x)$:

\[
\Pr\left[
\Ver\big(\crs, x, \pi\big) \ra 0
 \ :
\begin{array}{rl}
 &\crs \gets \Setup(1^\secp)\\
 & \pi \gets \Pro(\crs, x, \omega)\\
\end{array}
\right] \le \negl
\]

\item \emph{Statistical (Non-Adaptive) Soundness:} For every $\tx \notin \cL$, the following holds:

\[
\Pr_{\crs \la \Setup(1^\secp)}\left[
    \exists \tpi : \Ver(\crs, \tx, \tpi) \ra 1
\right] \le \negl.
\]

\item \emph{Statistical Adaptive Soundness:} 

\[
\Pr_{\crs \la \Setup(1^\secp)}\left[
\exists (\tx, \tpi) : \tx \notin \cL \land \Ver(\crs, \tx, \tpi) \ra 1
\right] \le \negl.
\]

\item \emph{Computational (Non-Adaptive) Zero-Knowledge with
Statistical CRS Indistinguishability:} There exists
a QPT simulator
$\Sim$ such that for all $x \in \cL$ and $\wit \in \RL(x)$, the
following holds wrt non-uniform QPT distinguishers:

\[
\left[
\begin{array}{ll}
(x, \crs, \pi) \approx_c (x, \crs', \pi')\\
\;\;\;\;\;\;\;\;\;\;\;\;\;\;\;\;\land \\
\;\;\;\;(x, \crs) \approx_s (x, \crs')
\end{array}
 \ :
\begin{array}{rl}
 &\crs \la \Setup(1^\secp)\\
 & \pi \la \Pro(\crs, x, \wit)\\
 & (\crs', \pi') \la \Sim(x)\\
\end{array}
\right]
\]

\begin{remark}
Statistical CRS
indistinguishability cannot hold for the CRS output by $\Sone$ of an
adaptive ZK simulator $(\Sone, \Stwo)$ of a
statistically-sound NIZK. This was
mentioned in \cite{AC:Groth06}, and can be seen from the proof of
\cite{TCC:Pass13}.
\end{remark}

\item \emph{Superposition-Secure (Adaptive) Zero-Knowledge:} There exists a
PPT two-part simulator $(\Sone, \Stwo)$ such that the following
holds for every QPT adversary $\qA$ that only queries \emph{valid}
statement-witness pairs in superposition, i.e., terms of the form
$\ket{x, \wit}$ where $x \in \cL \land \wit \in \RL(x)$:

\[
\left[
\begin{array}{ll}
\qA^{\ket{\Stwo(\td, \cdot)}}(\crs')\\
\;\;\;\;\;\;\;\approx_c\\
\qA^{\ket{\Pro(\crs, \cdot, \cdot)}}(\crs)\\
\end{array}
 \ :
\begin{array}{rl}
 &\crs \la \Setup(1^\secp)\\
 & (\crs', \td) \la \Sone(1^\secp)\\
\end{array}
\right]
\]

\end{itemize}
\end{definition}


\begin{theorem}[\cite{EC:QuaRotWic19,EC:WatWeeWu25,CiC:ACEMV24}]
Assuming the polynomial hardness of the LWE
problem, there exist NIZKs for $\NP$ in the CRS model with
the following properties:
\begin{itemize}
\item Statistical (Adaptive) Soundness
\item Superposition-Secure (Adaptive) Zero-Knowledge
\item Computational (Non-Adaptive) ZK with Statistical CRS Indistinguishability
\end{itemize}
\end{theorem}

In more detail, the work of \cite{EC:QuaRotWic19} showed that any
statistical-binding hidden bit generator (HBG) (Definition
\ref{def:hbg}) implies a statistically-sound NIZK. The works of
\cite{STOC:Waters24,EC:WatWeeWu25} constructed such HBGs from LWE with
sub-exponential and polynomial modulus respectively. In the work of
\cite{CiC:ACEMV24}, it was argued that the NIZK of \cite{TCC:BitPan15}
from iO and OWFs satisfies the stronger notion of superposition-secure
(adaptive) zero-knowledge. We observe that their reasoning also
applies to the NIZK of \cite{EC:WatWeeWu25} and the compiler of
\cite{EC:QuaRotWic19} due to the following reason:

The high level intuition is that these constructions simulate the idealized
hidden bits model (HBM) \cite{FLS90} in the CRS model. Note that the
idealized HBM provides \emph{perfect} zero-knowledge \cite{FLS99,AC:HofUrs19} (along with
statistical soundness), and hence also superposition-secure zero-knowledge.
The observation made by \cite{CiC:ACEMV24} is that the HBM to CRS
model transformation also plays well with superposition queries. This
is because the proof of the zero-knowledge property only requires
switching the CRS distribution computationally, before appealing to
the \emph{perfect} zero-knowledge guarantee of NIZK in the HBM. Since the CRS is classical and
independent of the statement-witness queries, the quantum-security implication follows. For
more details, see Appendix D of \cite{CiC:ACEMV24}.

Finally, we observe that in the HBG to NIZK compiler of
\cite{EC:QuaRotWic19}, the CRS in the construction is computed as
$(\crs_\bg, s)$ where the former is the CRS of the HBG and $s$ is a
truly random string. In the simulation $\Sim(x)$, $\crs_\bg$ is
computed the same way, but $s$ is computed as $s \seteq r_\bg \oplus
r_\hb$ where $r_\bg$ is a pseudo-random string generated by the HBG,
while $r_\hb$ is a uniform hidden-bit string output by the hidden
bits simulator $\Sim_\hb(x)$. Hence, statistical CRS
indistinguishability follows by the perfect zero-knowledge guarantee
of the hidden bits NIZK proof.

\subsection{Certified Deletion of BB84 States}\label{sec:bb84}

\begin{theorem}[Certified Deletion of BB84
States \cite{C:BarKhu23}]\label{thm:cert-del}
Consider a distribution $\cZ(\cdot, \cdot, \cdot)$
with three arguments: 1) a $\secp$-bit string $\theta$, 2) a bit $b$,
and 3) a $\secp$-qubit register $\qreg{D}$. For all $\theta \in
\bit^\secp$, $b' \in \bit$ and quantum states
$\ket{\psi}_{\qreg{D}\otimes\qreg{C}}$ on a $\secp$-qubit register
$\qreg{D}$ and an arbitrary size register $\qreg{C}$, let the
following hold for all QPT distinguishers $\qD:$

$$\bigg\lvert \Pr\Big[\qD\left(1^\secp, \cZ(\theta, b', \qreg{D}),
    \qreg{C}\right) = 1\Big] -
\Pr\Big[\qD\left(1^\secp, \cZ(0^\secp, b', \qreg{D}),
    \qreg{C}\right) = 1\Big]
\bigg\rvert \le \negl$$

In other words, $\cZ$ is semantically-secure wrt $\theta$. Consider
now the following experiment wrt an adversary $\qA$:

$\underline{\Exp_{\cZ, \qA}(b)}:$

\begin{enumerate}
\item The experiment samples $y, \theta \la \bit^\secp$ and
initializes $\qA$ with input $1^\secp$ along with the following
input:
$$\cZ\bigg(\theta, b \oplus \bigoplus_{i:\theta_i = 0}y_i,
\ket{y}^\theta\bigg)$$
\item $\qA$ sends a string $\cert \in \bit^\secp$ and a
quantum state $\rhoB$ on register $\qreg{B}$.
\item If $\forall i \in [\secp]$ such that $\theta_i = 1$, $\cert_i =
y_i$, then the experiment outputs $\rhoB$. Else, it outputs $\bot$.
\end{enumerate}

Then, the following guarantee holds for all QPT adversaries $\qA:$
$$\TD\bigg(\Exp_{\cZ, \qA}(0), \Exp_{\cZ, \qA}(1)\bigg) = \negl$$
\end{theorem}


\section{Certified-Everlasting NIZK Proofs }\label{sec:def_ce-niwi}

\begin{definition}[CE-NIZK in the CRS Model]\label{def:ce-nizk}
A certified-everlasting non-interactive zero-knowledge
proof (CE-NIZK) in the CRS model for a language $\cL$ is
a tuple of four algorithms $(\Setup, \Pro, \Ver, \allowbreak\DelVrfy)$ with
the following syntax:

\begin{itemize}
\item $\Setup(1^\secp) \to \crs:$ The classical setup algorithm takes
as input a security parameter, and outputs a common reference string (CRS) $\crs$.

\item $\Pro(\crs, x, \wit) \to (\sigma, \qP):$ The quantum prover's
algorithm takes as input a CRS $\crs$, a statement $x \in \cL$, and
a witness $\wit \in \RL(x)$. It outputs a quantum proof state $\sigma$ and a
residual state $\qP$.

\item $\Ver(\crs, x, \sigma) \to (b, \qV):$ The quantum verifier's
algorithm takes as input the CRS $\crs$, a statement $x \in \cL$, and
a quantum proof $\sigma$. It outputs a bit $b$ along with a residual state $\rho_{\Ver}$.

\item $\DelVrfy(\qP, \qV) \to \top / \bot:$ The
certification algorithm takes as input a prover's state
$\qP$ and a verifier's state $\qV$. It outputs $\top$ (accept) or $\bot$ (reject).
\end{itemize}

A CE-NIZK must satisfy the following properties:

\begin{itemize}
\item \emph{Completeness:} The following holds for every $x \in \cL$ and
$\wit \in \RL(x)$:

\[
\Pr\left[
\Ver\big(\crs,x,\sigma\big)[1] \ra 0
 \ :
\begin{array}{rl}
 &\crs \gets \Setup(1^\secp)\\
 & (\sigma, \qP) \gets \Pro(\crs, x, \omega)\\
\end{array}
\right] \le \negl
\]

\item \emph{Deletion Correctness:} The following holds for every $x \in \cL$
and $\wit \in \RL(x)$:

\[
\Pr\left[
\DelVrfy\big(\qP,\qV\big) \ra \bot
 \ :
\begin{array}{rl}
 &\crs \la \Setup(1^\secp)\\
 & (\sigma, \qP) \la \Pro(\crs, x, \wit)\\
 & (b, \qV) \la \Ver(\crs, x, \sigma)
\end{array}
\right] \le \negl
\]

\item \emph{Statistical Soundness:} For every $\tx \notin \cL$, the following holds:

\[
\Pr_{\crs \la \Setup(1^\secp)}\left[
    \exists \tsigma : \Ver(\crs, \tx, \tsigma)[1] \ra 1
\right] \le \negl.
\]

%

%
%
%
%

\item \emph{Computational Zero-Knowledge:} There exists a simulator
$\Sim$ such that for all $x \in \cL$ and $\wit \in \RL(x)$, the
following holds:

\[
\left[
(x, \crs, \sigma) \approx_c (x, \crs', \sigma')
 \ :
\begin{array}{rl}
 &\crs \la \Setup(1^\secp)\\
 & (\sigma, \qP) \la \Pro(\crs, x, \wit)\\
 & (\crs', \sigma') \la \Sim(x)\\
\end{array}
\right]
\]

\item \emph{Certified-Everlasting Zero-Knowledge (CE-ZK):}
For every malicious QPT verifier $\Ver^*$, there exists a QPT simulator
$\Sim_{\Ver^*}$ such that the following holds for every $x \in \cL$
and $\wit \in \RL(x)$:

\[
\left[
\trho \approx_s  \rho'
 \ :
\begin{array}{rl}
 &\crs \la \Setup(1^\secp)\\
 & (\sigma, \qP) \la \Pro(\crs, x, \wit)\\
 & \rho_{\qreg{A}\otimes\qreg{B}} \la \Ver^*(\crs, x, \sigma)\\
 & \textrm{If}\;\DelVrfy(\qP,\rho_{\qreg{A}}) = \top \;
 \textrm{then}\; \trho
 \seteq \rho_{\qreg{B}}.\\
 & \textrm{Else}\; \trho \seteq \bot.\\
 & \rho' \la \Sim_{\Ver^*}(x)
\end{array}
\right]
\]

Here, $\rho_{\qreg{A}\otimes\qreg{B}}$ is a possibly entangled state
on two registers $\qreg{A}$ and $\qreg{B}$. We denote the residual
state on $\qreg{A}$ (likewise $\qreg{B}$) at any point as
$\rho_{\qreg{A}}$ (likewise $\rho_{\qreg{B}}$).

\begin{remark}
The CE-ZK definition describes non-black-box simulation (for
generality in Theorem \ref{thm:del_res_imp}), but our feasibility
uses only black-box simulation.
\end{remark}

%
%
%
%
%
%
%
%
%
\end{itemize}

\end{definition}


\begin{definition}[CE-NIZK in the Shared EPR Model] \label{def:ce_nizk_epr}
Let $\qreg{(P_1, V_1)} \ldots \allowbreak\qreg{(P_\ell, V_\ell)}$ be registers
denoting halves of $\ell$-many EPR pairs. Let the prover's register be
defined as $\qreg{P} \seteq \qreg{P_1} \otimes \ldots \otimes \qreg{P_\ell}$ and the
verifier's register as $\qreg{V} \seteq \qreg{V_1} \otimes \ldots
\otimes \qreg{V_\ell}$. A CE-NIZK in the shared EPR model for a
language $\cL \in \NP$ is a tuple of five algorithms $(\Setup, \Pro,
\Ver, \Del, \DelVrfy)$ with the following syntax:

\begin{itemize}
\item $\Setup(1^\secp) \to \crs:$ The classical setup algorithm takes
as input a security parameter, and outputs a common reference string (CRS) $\crs$.

\item $\Pro(\crs, \qreg{P}, x, \wit) \to (\pi, \qP):$ The quantum prover's
algorithm takes as input a CRS $\crs$, the prover register
$\qreg{P}$, a statement $x \in \cL$, and
a witness $\wit \in \RL(x)$. It outputs a classical proof $\pi$ and a
residual state $\qP$ on a register $\qreg{P'} \otimes \qreg{P}$ where
$\qreg{P}$ is the EPR register and $\qreg{P'}$ is a new register.

\item $\Ver(\crs, \qreg{V}, x, \pi) \to (b, \qV):$ The quantum verifier's
algorithm takes as input a CRS $\crs$, the verifier register
$\qreg{V}$, a statement $x \in \cL$, and
a classical proof $\pi$. It outputs a bit $b$ along with a residual
state $\rho_{\Ver}$ on register $\qreg{V'}\otimes\qreg{V}$ where
$\qreg{V}$ is the EPR register and $\qreg{V'}$ is a new register.

\item $\Del(\qV) \to (\cert, \tqV):$ The quantum
deletion algorithm takes as input a state $\rho_\Ver$ on register
$\qreg{V'}\otimes\qreg{V}$, and outputs a classical certificate
$\cert$, and state $\tqV$.

\item $\DelVrfy(\cert, \qP) \to \top / \bot:$ The
certification algorithm takes as input a classical certificate
$\cert$ and a quantum state $\rho_\Pro$ on $\qreg{P'} \otimes
\qreg{P}$. It outputs $\top$ or $\bot$.
\end{itemize}

A CE-NIZK in the shared EPR model must satisfy the following properties:

\begin{itemize}
\item \emph{Completeness:} The following holds for every $x \in \cL$ and
$\wit \in \RL(x)$:

\[
\Pr\left[
\Ver\big(\crs,\qreg{V},x,\pi\big)[1] \ra 0
 \ :
\begin{array}{rl}
 &\crs \gets \Setup(1^\secp)\\
 & (\pi, \qP) \gets \Pro(\crs, \qreg{P}, x, \omega)\\
\end{array}
\right] \le \negl
\]

\item \emph{Deletion Correctness:} The following holds for every $x \in \cL$
and $\wit \in \RL(x)$:

\[
\Pr\left[
\DelVrfy\big(\cert,\qP\big) \ra \bot
 \ :
\begin{array}{rl}
 &\crs \la \Setup(1^\secp)\\
 & (\pi, \qP) \la \Pro(\crs, \qreg{P},x, \wit)\\
 & (b, \qV) \la \Ver(\crs, \qreg{V}, x, \pi)\\
 & (\cert, \tqV) \la \Del(\qV)
\end{array}
\right] \le \negl
\]

\item \emph{Statistical Soundness:} For every $\tx \notin \cL$ and
every unbounded $\Pro^*$, we have:

\[
\Pr\left[
\Ver\big(\crs,\qreg{V},\tx,\pi^*)[1] \ra 1
 \ :
\begin{array}{rl}
 &\crs \la \Setup(1^\secp)\\
 & (\pi^*, \qP^*) \la \Pro^*(\crs,\qreg{P},\tx)\\
\end{array}
\right] \le \negl
\]

\item \emph{Computational Zero-Knowledge:} There exists a simulator
$\Sim$ such that for all $x \in \cL$ and $\wit \in \RL(x)$, the
following holds:

\[
\left[
(x, \crs, \qreg{V}, \pi) \approx_c (x, \crs', \qreg{V'}, \pi')
 \ :
\begin{array}{rl}
 &\crs \la \Setup(1^\secp)\\
 & (\pi, \qP) \la \Pro(\crs, \qreg{P}, x, \wit)\\
 & (\crs', \qreg{V'}, \pi') \la \Sim(x)\\
\end{array}
\right]
\]

\item \emph{Certified-Everlasting Zero-Knowledge (CE-ZK):}
There exists a QPT simulator $\Sim$ such that for every malicious
verifier $\Ver^*$, statement $x \in \cL$ and witness $\wit \in
\RL(x)$, the following holds:

\[
\left[
 \trho \approx_s \rho'
 \ :
\begin{array}{rl}
 &\crs \la \Setup(1^\secp)\\
 & (\pi, \qP) \la \Pro(\crs, \qreg{P}, x, \wit)\\
 & (\cert^*, \rho^*) \la \Ver^*(\crs, \qreg{V}, x, \pi)\\
 & \textrm{If}\;\DelVrfy(\cert^*, \qP) = \top \; \textrm{then}\; \trho
 \seteq \rho^\star.\\
 & \textrm{Else}\; \trho \seteq \bot.\\
 & \rho' \la \Sim^{\Ver^*}(x)
\end{array}
\right]
\]

Here, $\rho^*$ is a state on $\qreg{V} \otimes
\qreg{V'}$ where $\qreg{V'}$ is a
new register, and $\Sim^{\Ver^*}$ denotes black-box access to
$\Ver^*$ on its input and setup registers.

%
%
%
%
%
%
%
%
%
\end{itemize}
\end{definition}


\section{CE-NIZK in the CRS Model}\label{sec:crs}

\subsection{Impossibility of Deletion-Resistant
CE-NIZKs}\label{sec:imp}

\begin{theorem}
If there exists a deletion-resistant CE-NIZK (Definition
\ref{def:dr-ce-nizk}) for a language $\cL$,
then there exists a quantum NIZK for $\cL$ satisfying
statistical soundness and statistical zero-knowledge.
\end{theorem}

\begin{proof}
Let $\sPi$ be a deletion-resistant CE-NIZK for a language $\cL$. This
means it has algorithms $(\Setup, \Pro, \Ver, \DelVrfy)$. Furthermore,
there exist QPT algorithms $\Ver^*, \tVer$ satisfying the conditions
of the theorem by assumption.  We now define the following prover
$\tPro$. Note that we can assume $\tPro$ to be initialized with
polynomially-many copies of the initial quantum states of $\Pro$ and $\Ver^*$.

\begin{description}
\item $\underline{\tPro(\crs, x, \wit)}$:

\begin{enumerate}
\item Compute $(\sigma, \qP) \la \Pro(\crs, x, \wit)$.
\item Compute $\rhoAB \la \Ver^*(\sigma)$.
\item If $\DelVrfy(\qP, \rhoA) = \top$, then output $\rhoB$.
Otherwise, repeat from Step 1.
\end{enumerate}
\end{description}

Observe now that the NIZK proof $\tsPi \seteq (\Setup, \tPro, \tVer)$ is
statistically-sound due to condition (3) of Definition
\ref{def:dr-ce-nizk}. Also,
completeness follows from the fact that $\tPro$ will obtain $\rhoAB$
such that $\DelVrfy(\qP, \rhoA) = \top$ with overwhelming probability
in polynomial time (due to condition (1) of Definition
\ref{def:dr-ce-nizk}), and
the fact that $\tVer$ accepts $\rhoB$ from condition (2). Finally,
observe that by viewing $\Ver^*$ as a malicious QPT verifier attacking
the CE-ZK property of $\sPi$, the state $\rhoB$ output by $\tPro$ can
be statistically simulated. This means that $\tsPi$ satisfies
statistical zero-knowledge.
\end{proof}

\subsection{Construction from LWE}

We now construct a CE-NIZK $\sPi = (\Setup, \Pro, \Ver, \DelVrfy)$ as
follows:
\begin{framed}
\begin{center}
\textbf{Certified-Everlasting NIZK in the CRS Model}
\end{center}
\noindent \underline{Building Blocks:}
\begin{enumerate}
\item NIZK $\sPi_\nizk \seteq (\Setup_\nizk, \Pro_\nizk, \Ver_\nizk)$
for $\NP$ in the CRS model with:
\begin{itemize}
\item Statistical Adaptive Soundness
\item Superposition-Secure Adaptive Zero-Knowledge
\item Non-Adaptive ZK with Statistical CRS Indistinguishability
\end{itemize}

\item BB84 States $\ket{y}^\theta \seteq
    H^{\theta_1}\ket{y_1} \otimes \ldots \otimes
H^{\theta_\secp}\ket{y_\secp}$, where $y, \theta \in \bit^\secp$.

\item Quantum-Secure PRF $\{\F_\prfk: \bit^{\poly(\secp)} \ra
\bit^{\poly(\secp)}\}_{\prfk \in \bit^\secp}$.
\end{enumerate}

\noindent $\underline{\Setup(1^\secp):}$
\begin{itemize}
\item Compute $\crs_\inn \la \Setup_\nizk(1^\secp)$ and $\crs_\out \la \Setup_\nizk(1^\secp)$.
\item Output $\crs \seteq \crs_\inn \| \crs_\out$.
\end{itemize}

\noindent $\underline{\Pro(\crs, x, \wit):}$
\begin{enumerate}
\item Let $\ell$ be an upper bound on the proof size of $\sPi_\nizk$.
For each $i \in [2\ell]$, sample $y^i, \theta^i \la \bit^\secp$. Then,
construct the following BB84 state on register $\qreg{R^i} \seteq
\qreg{R^i_1 \otimes \ldots \otimes R^i_\secp}$.

$$\ket{y^i}^{\theta^i}_{\qreg{R^i}} \seteq
H^{\theta^i_1}\ket{y^i_1}_{\qreg{R^i_1}} \otimes
\ldots \otimes H^{\theta^i_\secp}\ket{y^i_\secp}_{\qreg{R^i_\secp}}$$

\item Parse $\crs = \crs_\inn \| \crs_\out$. Then, compute $\pi_\inn \la \Pro_\nizk(\crs_\inn, x, \wit)$.

\item Sample $\key^0, \key^1 \la \bit^\ell$ and compute $\ct^0 \seteq
\pi_\inn \oplus \pad^0 \oplus \key^0$ and $\ct^1 \seteq 0^\ell \oplus
\pad^1 \oplus \key^1$, where:

$$\pad^1 \seteq \bigoplus\limits_{i:\theta^1_i=0}y^1_i \; \; \| \ldots \|
\bigoplus\limits_{i:\theta^\ell_i=0}y^\ell_i;\;\;
\pad^0 \seteq \bigoplus\limits_{i:\theta^{\ell+1}_i=0}y^{\ell+1}_i \; \; \| \ldots \|
\bigoplus\limits_{i:\theta^{2\ell}_i=0}y^{2\ell}_i$$

\item For each $\sz \in \bit^{2\ell\secp}$, let $\sz \seteq \sz^1 \|
\ldots \| \sz^{2\ell}$ where $|\sz^i| = \secp$ for each $i \in
[2\ell]$. Then, define the following $\NP$ statement $x^{\sz}$:

\underline{Statement $x^\sz$}:
There exist $\{\theta^i\}_{i\in[2\ell]}$
and $\key^0, \key^1 \in \bit^{\ell}$ s.t. $$\Ver_\nizk\left(\crs_\inn, x,
\ct^0 \oplus \key^0 \oplus \tpad^0\right) = 1 \bigvee
\Ver_\nizk\left(\crs_\inn, x, \ct^1 \oplus \key^1 \oplus \tpad^1\right) = 1$$ 
$$\tpad^1 \seteq \bigoplus\limits_{i:\theta^1_i=0}\sz^1_i \; \; \| \ldots \|
\bigoplus\limits_{i:\theta^\ell_i=0}\sz^\ell_i;\;\;
\tpad^0 \seteq \bigoplus\limits_{i:\theta^{\ell+1}_i=0}\sz^{\ell+1}_i \; \; \| \ldots \|
\bigoplus\limits_{i:\theta^{2\ell}_i=0}\sz^{2\ell}_i$$


\item For each $i \in [2\ell\secp]$, sample $s^{i,0}, s^{i,1}
\la \bit^\secp$. For each $\sz \in \bit^{2\ell\secp}$:
$$\sig^\sz \seteq s^{1, \sz[1]} \| \ldots \| s^{2\ell\secp,
\sz[2\ell\secp]}$$

\item Define the register $\qreg{R} \seteq \qreg{R^1} \otimes \ldots
\otimes \qreg{R^{2\ell}}$. Let the state on register $\qreg{R}$ be
$\sum\limits_{\sz}\alpha^{\sz}\ket{\sz}$. Then, compute the following
state $\ket{\psi}$ on registers $\qreg{R}$, $\qreg{P}$ and $\qreg{S}$:

$$\ket{\psi}_{\qreg{R} \otimes \qreg{P} \otimes \qreg{S}} \seteq \sum_{\sz} \alpha^\sz
\ket{\sz}_{\qreg{R}}
\ket{\pi_{\out}^{\sz}}_{\qreg{P}}\ket{\sig^\sz}_{\qreg{S}}$$

where $\pi_\out^\sz \la
\Pro_\nizk(\crs_\out, x^\sz, \wit_\out; r^\sz)$, 
$\wit_\out \seteq (\{\theta^i\}_{i\in[2\ell]},
\key^0, \key^1)$ and explicit randomness $r^\sz \seteq \F_\prfk(\sz)$ for a PRF key sampled as $\prfk \la \bit^\secp$.

\item Output the quantum proof $\sigma \seteq
\left(\ket{\psi}_{\qreg{R}\otimes \qreg{P}\otimes \qreg{S}}, \ct^0, \ct^1\right)$
and the verification state $\qP \seteq \left(\wit_\out, \crs,
\{y^i\}_{i\in[2\ell]}, \prfk, \{s^{i,0}, s^{i,1}\}_{i\in[2\ell\secp]},
\ct^0, \ct^1\right)$.
\end{enumerate}

\noindent $\underline{\Ver(\crs, x, \sigma):}$
\begin{enumerate}
\item Parse $\crs$ as $\crs = \crs_\inn \| \crs_\out$ and $\sigma = (\ket{\psi}_{\qreg{R} \otimes \qreg{P} \otimes \qreg{S}}, \ct^0, \ct^1)$.

\item Let $\ket{\psi}_{\qreg{R \otimes P} \otimes \qreg{S}} = \sum\limits_\sz
\alpha^\sz
\ket{\sz}_{\qreg{R}}\ket{\pi_\out^\sz}_{\qreg{P}}\ket{\sig^\sz}_{\qreg{S}}$.
Compute the following state on registers $\qreg{R}, \qreg{P},
\qreg{S}, \qreg{OUT}$:

$$\ket{\psi}'_{\qreg{R} \otimes \qreg{P} \otimes \qreg{S} \otimes
\qreg{OUT}} \seteq \sum\limits_\sz
\alpha^\sz
\ket{\sz}_{\qreg{R}}\ket{\pi^\sz_\out}_{\qreg{P}}\ket{\sig^\sz}_{\qreg{S}}\ket{\Ver_\nizk(\crs_\out,
x^\sz, \pi_\out^\sz)}_{\qreg{OUT}}$$
\item Measure the register $\qreg{OUT}$ in the computational basis to
obtain outcome $b \in \bit$ and a residual state
$\ket{\widetilde{\psi}}_{\qreg{R} \otimes \qreg{P} \otimes \qreg{S}}$.
\item Set $\qV \seteq
    \ket{\widetilde{\psi}}_{\qreg{R} \otimes \qreg{P} \otimes \qreg{S}}$. Finally, output $(b, \qV)$.
\end{enumerate}

\noindent $\underline{\DelVrfy(\qP, \qV):}$

\begin{enumerate}
\item Parse
$\qP = \left(\wit_\out, \crs,
\{y^i\}_{i\in[2\ell]}, \prfk, \{s^{i,0}, s^{i,1}\}_{i\in[2\ell\secp]},
\ct^0, \ct^1\right)$ and $\qV =
\ket{\tpsi}_{\qreg{R}\otimes\qreg{P}\otimes\qreg{S}}$.

\item Perform $\ket{\sz}_{\qreg{R}}\ket{\tsig^\sz}_{\qreg{S}}\ket{0}_{\qreg{T}} \mapsto
\ket{\sz}_{\qreg{R}}\ket{\tsig^\sz}_{\qreg{S}}\ket{0 \oplus \Test(\sz,
\tsig^\sz)}_{\qreg{T}}$ where $\qreg{T}$ is initialized to $0$ and
$\Test(\sz, \tsig^\sz)$ is a function that outputs 1 if $\sig^\sz =
\tsig^\sz$ and $0$ otherwise. Measure register $\qreg{T}$ and output $\bot$ if the
outcome is $0$.

\item Perform $\ket{\sz}_{\qreg{R}}\ket{\sft}_{\qreg{P}}\ket{\sfs}_{\qreg{S}}
\mapsto \ket{\sz}_{\qreg{R}}\ket{\sft \oplus
\Pro_\nizk(\crs_\out, x^{\sz}, \wit_\out;
r^{\sz})}_{\qreg{P}}\ket{\sfs \oplus \sig^\sz}_{\qreg{S}}$
where $r^\sz = \F_\prfk(\sz)$ to get the state
$\ket{\tphi}_{\qreg{R} \otimes \qreg{P} \otimes \qreg{S}}$.

\item Measure the $\qreg{R}$ register of $\ket{\tphi}_{\qreg{R}
\otimes \qreg{P} \otimes \qreg{S}}$ in the Hadamard basis to get $\cert$.

\item Let $\theta \seteq \theta^1 \| \ldots \| \theta^{2\ell}$ and $y
\seteq y^1 \| \ldots \| y^{2\ell}$. If $\cert_j =
y_j$ for each $j \in [2\ell\secp]$ such that $\theta_j = 1$, then
output $\top$. Otherwise, output $\bot$.
\end{enumerate}
\end{framed}

\begin{theorem} \label{thm:crs}
Assuming $\sPi_\nizk$ is a NIZK proof for $\NP$ with statistical
adaptive soundness, superposition-secure (adaptive)
zero-knowledge and computational (non-adaptive)
zero-knowledge with statistical CRS indistinguishability (Definition
\ref{def:nizk}), and $\{\F_\prfk\}_\prfk$ is a quantum-secure PRF
family, then $\sPi \seteq (\Setup, \Pro,
\Ver, \allowbreak \DelVrfy)$ is a secure CE-NIZK scheme for
$\NP$ (Definition \ref{def:ce-nizk}).
\end{theorem}

\begin{remark}
Definition \ref{def:ce-nizk} only enforces non-adaptive soundness and 
computational ZK. However, it is easy to see that the stronger adaptive notions
follow for $\sPi$ from those of the underlying NIZK $\sPi_\nizk$, by nature of the
construction.
\end{remark}

\begin{proof}

\paragraph{Completeness:} Observe that
$\Ver_\nizk(\crs_\inn,
x, \ct^0 \oplus \tpad^0 \oplus \key^0)$ outputs 1 for each
superposition term $\sz$. This is because
unmasking $\ct^0$ with $\key^0$ will result in $\pi_\inn
\oplus \pad^0$. Moreover, $\tpad^0$ corresponding to every
superposition term $\sz$ is equal to $\pad^0$. Next, observe that
$\Ver_\nizk(\crs_\inn, x, \pi_\inn)$ outputs $1$ with overwhelming
probability by the completeness of $\Pi_\nizk$. Consequently, with all
but negligible probability, all the
statements $x^\sz$ in superposition are valid. Finally, by completeness of $\Pi_\nizk$, we have
that $\Ver_\nizk(\crs_\out, x^\sz, \pi_\out^\sz)$ outputs 1 with
overwhelming probability for each term $\sz$, assuming the randomness
for computing $\pi^\sz_\out$ is obtained from a quantum-secure PRF
$\F_\prfk$. Consequently, computing the result of verification on
register $\qreg{OUT}$ and measuring the result produces outcome $1$
with overwhelming probability.

\paragraph{Deletion-Correctness:} Recall that
$\Ver$ is sent the state
$\ket{\psi}_{\qreg{R}\otimes\qreg{P}\otimes\qreg{S}} =
\sum_\sz \alpha^\sz
\ket{\sz}_{\qreg{R}}\allowbreak\ket{\pi_\out^\sz}_{\qreg{P}}\ket{\sig^\sz}_{\qreg{S}}$
and the post-verification state
$\ket{\tpsi}$ is close in trace distance to
$\ket{\psi}$ due to the above completeness argument and the gentle
measurement lemma. Next, observe that the signature checks using
$\Test$ pass in superposition for every term $\sz$ and hence the state
is undisturbed. Then,
$\DelVrfy$ uncomputes the proofs on registers $\qreg{P}, \qreg{S}$, resulting in the state
$\ket{\tphi}$ that is close in trace distance to the state
$\ket{y}^\theta$. Consequently, measuring the registers of
$\ket{\tphi}$ in the Hadamard basis results in a value $\cert$ that is
consistent with $y$ at all the Hadamard basis positions with
overwhelming probability. As a result,
$\DelVrfy$ outputs $1$ with overwhelming probability.

\paragraph{Computational Zero-Knowledge:} This follows immediately from the
fact that $\Pi_\nizk$ has a simulator $\Sim_\nizk$ such that
the values $(\crs_\Sim, \pi_\Sim) \la \Sim_\nizk(x)$ are computationally
indistinguishable from $\crs_\inn, \pi_\inn$. Consequently, we can
replace $\pi_\inn$ inside $\ct^0$ with $\pi_\Sim$ and $\crs_\inn$ (which
is part of $\crs$) with $\crs_\Sim$.

\paragraph{Statistical Soundness:} For every false statement $x'$,
with overwhelming probability over the choice of $\crs_\inn$, we know
that there doesn't exist any $\pi'$ such that $\Ver_\nizk(\crs_\inn,
x', \pi') = 1$ due to the statistical soundness of $\Pi_\nizk$.
Consequently, with overwhelming probability over $\crs_\inn$, all the
corresponding statements $x_\sz$ in superposition are false as well.
Now, these statements are decided after $\crs_\out$ is chosen.
Consequently, we rely on the adaptive soundness property to ensure
that with overwhelming probability over $\crs_\out$, there do not
exist pairs $(x^\sz, \pi_\out^\sz)$ that would be accepted by
$\Ver_\nizk(\crs_\out, x^\sz, \pi_\out^\sz)$. Consequently,
verification outputs $0$ with overwhelming probability, ensuring
statistical soundness. 

\paragraph{Certified-Everlasting Zero-Knowledge:}

We now proceed to prove the CE-ZK property using the
following sequence of hybrids:

$\underline{\Hyb_0:}$
\begin{itemize}
\item $\Ver^*$ is initialized with $x$ and $\crs = \crs_\inn \| \crs_\out \la \Setup(1^\secp)$.
\item Then, $\Ver^*$ is provided with $\sigma =
(\ket{\psi}_{\qreg{R}\otimes\qreg{P}\otimes\qreg{S}}, \ct^0, \ct^1)$
computed as $(\sigma, \qP) \la \Pro(\crs, x, \wit)$, where $\ct^0 =
\pi_\inn \oplus \pad^0 \oplus \key^0$ and $\ct^1 = 0^\ell \oplus
\pad^1 \oplus \key^1$.
\item Then, $\Ver^*$ outputs a bi-partite state $\rhoAB$. If
$\DelVrfy(\qP, \rhoA) = \top$, output $\rhoB$ where $\rhoA$ is the
state on register $\qreg{A}$ and $\rhoB$ is the residual state on
register $\qreg{B}$. Output $\bot$ otherwise.
\end{itemize}

$\underline{\Hyb_1:}$
\begin{itemize}
\item $\Ver^*$ is initialized with $x$ and
\textcolor{red}{$\crs_\Sim \| \crs_\out$}, where
\textcolor{red}{$(\crs_\Sim, \pi_\Sim) \la \Sim_\nizk(x)$} and
$\crs_\out \la \Setup_\nizk(1^\secp)$, with $\Sim_\nizk$ being the
QPT simulator of $\sPi_\nizk$.
\item Then, $\Ver^*$ is provided with $\sigma =
    (\ket{\psi}_{\qreg{R}\otimes \qreg{P}\otimes \qreg{S}}, \ct^0, \ct^1)$
computed as $(\sigma, \qP) \la \Pro(\textcolor{red}{\crs_\Sim \|
\crs_\out}, \allowbreak x, \wit)$, where $\ct^0 =
\pi_\inn \oplus \pad^0 \oplus \key^0$ and $\ct^1 =
0^\ell \oplus \pad^1 \oplus \key^1$.
\item Then, $\Ver^*$ outputs a bi-partite state $\rhoAB$. If
$\DelVrfy(\qP, \rhoA) = \top$, output $\rhoB$ where $\rhoA$ is the
state on register $\qreg{A}$ and $\rhoB$ is the residual state on
register $\qreg{B}$. Output $\bot$ otherwise.
\end{itemize}

$\Hyb_0 \approx_s \Hyb_1$ follows from the statistical CRS
indistinguishability property corresponding to $\Sim_\nizk$, which
ensures that $(x, \crs_\inn) \approx_s (x, \crs_\Sim)$.

Consider now the hybrids, $\Hyb_2^1, \ldots, \Hyb_2^\ell$ where
$\Hyb_2^i$ corresponds to changing the $i$-th bit of $0^\ell$ inside
$\ct^1$ into the $i$-th bit of $\pi_\Sim$. Note that this also changes
the statements $x^\sz$ for the superposition terms $\sz$
corresponding to the $\qreg{R}$ register of
$\ket{\psi}_{\qreg{R}\otimes\qreg{P}\otimes\qreg{S}}$. Consequently,
the hybrids also involve a change to the proofs
$\pi^\sz_\out$ in superposition on register $\qreg{P}$ accordingly. Now, let $\Hyb_2^0 \seteq
\Hyb_1$ and observe that $\Hyb_2^\ell$ is the following hybrid:

$\underline{\Hyb_2^\ell:}$
\begin{itemize}
\item $\Ver^*$ is initialized with $x$ and
\textcolor{black}{$\crs_\Sim \| \crs_\out$}, where
\textcolor{black}{$(\crs_\Sim, \pi_\Sim) \la \Sim_\nizk(x)$} and
$\crs_\out \la \Setup_\nizk(1^\secp)$, with $\Sim_\nizk$ being the
QPT simulator of $\sPi_\nizk$.
\item Then, $\Ver^*$ is provided with $\sigma =
(\ket{\psi}_{\qreg{R}\otimes \qreg{P} \otimes \qreg{S}}, \ct^0, \ct^1)$
computed as $(\sigma, \qP) \la \Pro(\textcolor{black}{\crs_\Sim \|
\crs_\out}, \allowbreak x, \wit)$, except that $\ct^0 =
\pi_\inn \oplus \pad^0 \oplus \key^0$ and $\ct^1 =
\textcolor{red}{\pi_\Sim} \oplus \pad^1 \oplus \key^1$.
\item Then, $\Ver^*$ outputs a bi-partite state $\rhoAB$. If
$\DelVrfy(\qP, \rhoA) = \top$, output $\rhoB$ where $\rhoA$ is the
state on register $\qreg{A}$ and $\rhoB$ is the residual state on
register $\qreg{B}$. Output $\bot$ otherwise.
\end{itemize}

\begin{lemma}\label{lemma:main}
$\forall i \in [\ell]: \Hyb_2^{i-1} \approx_s \Hyb_2^{i}$
\end{lemma}
\begin{proof}
Let us now focus on the $i$-th BB84 state
$\ket{y^i}^{\theta^i}_{\qreg{R^i}}$, and
the distribution $\cZ^i$:

$\underline{\cZ^i(\theta^i, \textcolor{red}{b'}, \qreg{D}):}$

\begin{itemize}
    \item If $\textcolor{red}{\theta^i \neq 0^\lambda}$, then execute the following:
\begin{enumerate}
\item Compute $\ct^0, \ct^1$ in the same way as in $\Hyb_2^{i-1}$
except $b'$ is set as the $i$-th bit inside $\ct^1$. In particular, we
have:
$$\ct^0 \seteq \pi_\inn \oplus \pad^0 \oplus \key^0, \;\; \ct^1 \seteq
\pi_\Sim[1] \| \ldots \| \pi_\Sim[i-1] \| \textcolor{red}{b'} \|
0^{\ell-i} \oplus \pad^1 \oplus \key^1$$

\item Compute $\ket{\psi}_{\qreg{R}\otimes\qreg{P}\otimes\qreg{S}}$ 
in the same way as in $\Hyb_2^{i-1}$, except using $\theta^i$ from the
input (instead of sampling it) and register $\qreg{D}$ instead of
sampling the state $\ket{y^i}^{\theta^i}$.

\item Initialize $\Ver^*$ with input $x, \crs_\Sim \| \crs_\out$ and provide it with the
proof $\sigma \seteq
(\ket{\psi}_{\qreg{R}\otimes\qreg{P}\otimes\qreg{S}}, \allowbreak \ct^0,
\ct^1)$.

\item When $\Ver^*$ outputs a bi-partite state $\rhoAB$, parse it is
as a state
$\rho_{\qreg{R}\otimes\qreg{P}\otimes\qreg{S}\otimes\qreg{B}}$ on registers
$\qreg{R}, \qreg{P}, \qreg{S}$ and $\qreg{B}$.

\item Perform the map
$\ket{\sz}_{\qreg{R}}\ket{\tsig^\sz}_{\qreg{S}}\ket{0}_{\qreg{T}} \mapsto
\ket{\sz}_{\qreg{R}}\ket{\tsig^\sz}_{\qreg{S}}\ket{0 \oplus \Test(\sz,
\tsig^\sz)}_{\qreg{T}}$
where $\Test$ is as defined in the construction. Then, measure register $\qreg{T}$ and 
output $\bot$ if the result is not $1$.

\item Perform $\ket{\sz}_{\qreg{R}}\ket{\sft}_{\qreg{P}}\ket{\sfs}_{\qreg{S}}
\mapsto \ket{\sz}_{\qreg{R}}\ket{\sft \oplus
\Pro_\nizk(\crs_\out, x^{\sz}, \wit_\out;
r^{\sz})}_{\qreg{P}}\ket{\sfs \oplus \sig^\sz}_{\qreg{S}}$ on
registers $\qreg{R}, \qreg{P}$ and $\qreg{S}$
where $r^\sz = \F_\prfk(\sz)$ to get the state
$\rho_{\qreg{R}\otimes\qreg{P}\otimes\qreg{S}\otimes\qreg{B}}'$.
\item Output the state
$\rho_{\qreg{R}\otimes\qreg{P}\otimes\qreg{S}\otimes\qreg{B}}'$ along
with $\{\theta^j, y^j\}_{j \in [2\ell] \setminus \{i\}}$.

\end{enumerate}
\item If $\textcolor{red}{\theta^i = 0^\lambda}$, then execute the following:
\begin{enumerate}
\item Compute the state $\sum\limits_{\sz} \alpha^\sz
\ket{\sz}_{\qreg{R}}$ in the same way as in $\Hyb_2^{i-1}$, except
that register $\qreg{D}$ is used instead of the state
$\ket{y^i}^{\theta^i}$.

\item Compute $\ct^0, \ct^1$ as random strings, i.e.:
$$\textcolor{red}{\ct^0 \la \bit^\ell, \;\;
\ct^1 \la \bit^\ell}$$
\item Then, compute the state
$\ket{\psi}_{\qreg{R}\otimes\qreg{P}\otimes\qreg{S}}$
as: 
$$\ket{\psi}_{\qreg{R}\otimes\qreg{P}\otimes\qreg{S}} \seteq
\sum_{\sz}\alpha^\sz\ket{\sz}_{\qreg{R}}\ket{\pi^\sz_{\sS}}_{\qreg{P}}\ket{\sig^\sz}_{\qreg{S}}$$
where the NIZK proof $\pi_{\sS}^\sz$ is computed as \textcolor{red}{$\pi_{\sS}^\sz
\la \Stwo(\td, x^\sz; r^\sz)$} where $r^\sz \seteq \F_\prfk(\sz)$ for
$\prfk \la \bit^\secp$ and $(\crs_{\sS}, \td) \la \Sone(1^\secp)$ where
$\sS \seteq (\Sone, \Stwo)$ is the superposition-secure adaptive
zero-knowledge simulator of $\sPi_\nizk$.

\item Initialize $\Ver^*$ with input $x, \crs_\Sim \|
\textcolor{red}{\crs_\sS}$ and
provide it with the proof $\sigma \seteq
(\ket{\psi}_{\qreg{R}\otimes\qreg{P}\otimes\qreg{S}}, \allowbreak\ct^0, \ct^1)$.

\item When $\Ver^*$ outputs a bi-partite state $\rhoAB$, parse it as a
state $\rho_{\qreg{R}\otimes\qreg{P}\otimes\qreg{S}\otimes\qreg{B}}$.

\item Perform the map
$\ket{\sz}_{\qreg{R}}\ket{\tsig^\sz}_{\qreg{S}}\ket{0}_{\qreg{T}} \mapsto
\ket{\sz}_{\qreg{R}}\ket{\tsig^\sz}_{\qreg{S}}\ket{0 \oplus \Test(\sz,
\tsig^\sz)}_{\qreg{T}}$
where $\Test$ is as defined in the construction. Then, measure register $\qreg{T}$ and 
output $\bot$ if the result is not $1$.

\item Perform the map
    $\ket{\sz}_{\qreg{R}}\ket{\sft}_{\qreg{P}}\ket{\sfs}_{\qreg{S}}
\mapsto \ket{\sz}_{\qreg{R}}\ket{\sft \oplus
\Stwo(\td, x^\sz; r^\sz)}_{\qreg{P}}\ket{\sfs \oplus \sig^\sz}_{\qreg{S}}$ on
registers $\qreg{R}, \qreg{P}$ and $\qreg{S}$
where $r^\sz = \F_\prfk(\sz)$ to get the state
$\rho_{\qreg{R}\otimes\qreg{P}\otimes\qreg{S}\otimes\qreg{B}}'$.
\item Output the state
$\rho_{\qreg{R}\otimes\qreg{P}\otimes\qreg{S}\otimes\qreg{B}}'$ along
with $\{\theta^j, y^j\}_{j \in [2\ell] \setminus \{i\}}$.
\end{enumerate}
\end{itemize}

\begin{lemma} \label{lemma:i}
For all $\theta^i \in \bit^\secp, b' \in \bit$ and states
$\ket{\tau}_{\qreg{C}\otimes\qreg{D}}$ on registers $\qreg{C},
\qreg{D}$ where $\qreg{D}$ is of size $\lambda$, the following holds:

$$\big(\qreg{C}, \cZ^i(\theta^i, b', \qreg{D})\big) \approx_c
\big(\qreg{C}, \cZ^i(0^\secp, b', \qreg{D})\big)$$
\end{lemma}

\begin{proof}
Consider the following hybrids that act on inputs $(\theta^i, b',
\qreg{D})$:
\begin{itemize}
\item $\underline{\cZ^i_0:}$ Same as $\cZ^i$.

\item $\underline{\cZ^i_1:}$ Similar to $\cZ^i_0$, except that after the signature
check in Step 5., another check is performed before Step 6., where the
terms $\sz$ in superposition are post-selected to be valid statements
wrt the witness $\wit_\out = (\{\theta^i\}_{i\in[2\ell]}, \key^0,
\key^1)$. In other words, the map
$\ket{\sz}_{\qreg{R}}\ket{0}_{\qreg{W}} \mapsto
\ket{\sz}_{\qreg{R}}\ket{\chk(\sz)}_{\qreg{W}}$ is first performed, where
$\chk(\sz)$ is a function that outputs $1$ iff $\wit_\out \in
\Rout(x^\sz)$, where $\Rout$ is the $\NP$ relation specified by the
description of $x^\sz$. Then, the $\qreg{W}$ register is measured and
$\bot$ is output if the result is not $1$.

\item $\underline{\cZ^i_2:}$ Similar to $\cZ_1^i$, except that it
generates simulated proofs $\pi^\sz_\sS$ on register $\qreg{P}$ as in
the distribution $\cZ(0^\secp, b', \qreg{D})$, instead of the real
proofs $\pi^\sz_\out$, apart from replacing $\crs_\out$ with a
simulated reference string $\crs_\sS$.

\item $\underline{\cZ^i_3:}$ Similar to $\cZ^i_2$, except that $\ct^0,
\ct^1$ are sampled as truly random values, and the simulated proofs
$\pi_\sS^\sz$ are generated according to the new ciphertexts.

\item $\underline{\cZ^i_4:}$ Similar to $\cZ^i_3$, except that the
post-selection step introduced in $\cZ^i_1$ is no longer performed.
\end{itemize}

We begin by arguing the indistinguishability of $\cZ_0^i$ and
$\cZ_1^i$ as follows:

\begin{lemma}\label{lemma:ii}
For all $\theta^i \in \bit^\secp$, $b' \in \bit$ and states
$\ket{\tau}_{\qreg{C}\otimes\qreg{D}}$ with a $\secp$-size register
$\qreg{D}$, we have:
$$(\qreg{C}, \cZ_0^i(\theta^i, b', \qreg{D}))
\approx_s (\qreg{C}, \cZ_1^i(\theta^i, b', \qreg{D}))$$
\end{lemma}

\begin{proof}
Recall first how the two distributions differ:

\begin{itemize}
\item $\underline{\cZ^i_0:}$ Same as $\cZ^i$.

\item $\underline{\cZ^i_1:}$ Similar to $\cZ^i_0$, except that after the signature
check in Step 5., another check is performed before Step 6., where the
terms $\sz$ in superposition are post-selected to be valid statements
wrt the witness $\wit_\out = (\{\theta^i\}_{i\in[2\ell]}, \key^0,
\key^1)$. In other words, the map
$\ket{\sz}_{\qreg{R}}\ket{0}_{\qreg{W}} \mapsto
\ket{\sz}_{\qreg{R}}\ket{\chk(\sz)}_{\qreg{W}}$ is first performed, where
$\chk(\sz)$ is a function that outputs $1$ iff $\wit_\out \in
\Rout(x^\sz)$, where $\Rout$ is the $\NP$ relation specified by the
description of $x^\sz$. Then, the $\qreg{W}$ register is measured and
$\bot$ is output if the result is not $1$.
\end{itemize}

Consider now the case when $\theta^i \neq 0^\secp$ and assume that
there exists an unbounded distinguisher $\D$ that has non-negligible
advantage in distinguishing between the distributions. Consider the
following functions $G, H$:

$\underline{G[\aux](\sz, \tsig^\sz):}$
\begin{itemize}
\item Output $(1, \{s^{i,0}, s^{i,1}\}_{i\in[2\ell\secp]})$ if $\Test(\sz, \tsig^\sz) = 1$ AND $\wit_\out \in \Rout(x^\sz)$, where $\Test, x^\sz, \wit_\out$, $\{s^{i,0},
s^{i,1}\}_{i\in[2\ell\secp]}$ and $\Rout$ are all as defined in
Section \ref{sec:crs}, and are parsed from $\aux$ that is part of the
description of $G$.
\item Otherwise, output $(0, \{s^{i,0}, s^{i,1}\}_{i\in[2\ell\secp]})$.
\end{itemize}

$\underline{H[\aux](\sz, \tsig^\sz):}$
\begin{itemize}
\item Output $(1, \{s^{i,0}, s^{i,1}\}_{i\in[2\ell\secp]})$ if
    $\Test(\sz, \tsig^\sz) = 1$ where $\Test, x^\sz, \wit_\out$,
    $\{s^{i,0}, \allowbreak
s^{i,1}\}_{i\in[2\ell\secp]}$ and $\Rout$ are all as defined in
Section \ref{sec:crs}, and are parsed from $\aux$ that is part of the
description of $H$.
\item Otherwise, output $(0, \{s^{i,0}, s^{i,1}\}_{i\in[2\ell\secp]})$.
\end{itemize}

We now construct an algorithm $\qA$ that makes use of
$\D$ in order to distinguish between the oracles $G$ and $H$.

$\underline{\qA^{\cO}(\rho_\aux):}$
\begin{itemize}
\item From the input $\rho_\aux$, obtain the state
$\rho_{\qreg{R}\otimes\qreg{P}\otimes\qreg{S}\otimes\qreg{B}}$, which
corresponds to the state computed after Step 4. of $\cZ^i$.

\item Evaluate $\cO$ (which is either $G$ or $H$) in superposition of registers
$\qreg{R}, \qreg{S}$ and post-select on the outcome being $1$. Also, obtain
$\{s^{i,0}, s^{i,1}\}_{i \in [2\ell\secp]}$ from $\cO$ and all
other values corresponding to $\cZ^i$ from the auxiliary input state $\rho_\aux$.
\item Perform the uncomputation step, i.e., Step 6. of $\cZ^i$ to get the state
$\rho'_{\qreg{R}\otimes\qreg{P}\otimes\qreg{S}\otimes\qreg{B}}$.
\item Finally, run $\D$ on input
$\rho'_{\qreg{R}\otimes\qreg{P}\otimes\qreg{S}\otimes\qreg{B}}$
and $\{\theta^j, y^j\}_{j\in[2\ell] \setminus \{i\}}$ and output the
bit output by $\D$.
\end{itemize}

Observe that $\qA^\cO$ provides $\D$ with either the output of $\cZ_0^i$ or
$\cZ^i_1$ depending on whether $\cO$ is $G$ or $H$.

Consider now the following algorithm $\qR$ which is given a state
$\rho_\aux$ which consists of the state
$\rho_{\qreg{R}\otimes\qreg{P}\otimes\qreg{S}\otimes\qreg{B}}$ that is
computed after Step 4. of $\cZ^i$ along with other values sampled in
$\cZ^i$. Importantly, $\qR$ is only given $s^{j, y[j]}$ for each $j
\in [(\ell+1)\secp, 2\ell\secp]$ such that $\theta_j = 0$. Then, $\qR$
runs the algorithm $\qA^H(\rho_\aux)$, except it doesn't respond to
the query and simply measures the oracle query made by $\qA$ to obtain
$(\tsz, \tsig^{\tsz})$. Observe that $\qR$ here simply corresponds to
the algorithm $\qB^H$ guaranteed by the O2H Lemma in the special case
when $q = 1$. Hence, with non-negligible probability, $\tsz$ is such
that $\wit_\out \notin \Rout(x^{\tsz})$ and yet $\tsig^{\tsz}$ is a
valid signature, i.e., $\tsig^{\tsz} = \sig^{\tsz}$. Clearly, this means there
exists some $j \in [(\ell+1)\secp, 2\ell\secp]$ where $\theta_j =
0$ and $\tsz[j] \neq y[j]$. Notice however, that $\qR$ was given no
information about $s^{j, 1 - y[j]}$. Consequently, it is easy to see
by a union bound analysis that $\tsig^{\tsz} = \sig^{\tsz}$ holds with
at most negligible probability, a contradiction. This means no such
distinguisher $\D$ exists, and the proof of the $\theta = 0^\secp$
case is essentially the same, finishing the proof of Lemma \ref{lemma:ii}.
\end{proof}


Next, we argue that $\cZ_1^i$ and $\cZ_2^i$ are computationally close.
Recall that by definition, the statements $x^\sz$ corresponding to superposition
terms $\sz$ are all true wrt
$\wit_\out$ before computation of
the proofs on register $\qreg{P}$. Moreover,
the same holds before uncomputation as well,
due to the post-selection step introduced in $\cZ^i_1$. Hence, we can
invoke superposition-secure adaptive zero-knowledge (Definition
\ref{def:nizk}) of $\sPi_\nizk$ to guarantee that $\big(\qreg{C},
\cZ^i_1(\theta^i, b', \qreg{D})\big) \approx_c \big(\qreg{C},
\cZ_2^i(\theta^i, b', \qreg{D})\big)$, as the two only differ in that
one uses real proofs, while the other uses simulated proofs. In
particular, both the computation and uncomputation of the proofs in
the former (likewise latter) distribution can be performed using
quantum oracle access (as in Definition \ref{def:nizk}) to
$\Pro_\nizk(\crs_\out, \cdot, \cdot)$ (likewise $\Stwo(\td, \cdot)$). Note
that by quantum-security of $\F_\prfk$, we can assume wlog that the oracle access
to $\Pro_\nizk(\crs_\out, \cdot, \cdot), \Stwo(\td, \cdot)$ in the
superposition-secure adaptive zero-knowledge definition to be simulated via
randomness derived from the PRF.

Observe now that $\big(\qreg{C},
\cZ^i_2(\theta^i, b', \qreg{D})\big) \equiv \big(\qreg{C},
\cZ_3^i(\theta^i, b', \qreg{D})\big)$ because in the former, the
ciphertexts $\ct^0, \ct^1$ are masked by uniform and independent values $\key^0,
\key^1$, while they are directly sampled from the uniform distribution
in the latter. The equivalence of the distributions then follows from
the fact that the simulated proofs are directly computed using the
statements $x^\sz$ (which depend on $\ct^0, \ct^1$).

Finally, $\big(\qreg{C},
\cZ^i_3(\theta^i, b', \qreg{D})\big) \approx_s \big(\qreg{C},
\cZ_4^i(\theta^i, b', \qreg{D})\big)$ follows from essentially the
same argument as that of Lemma \ref{lemma:ii}. Notice now that
$\cZ_4^i(\theta^i, b', \qreg{D})$ is the same distribution as 
$\cZ^i(0^\secp, b', \qreg{D})$, and recall that $\cZ^i_0(\theta^i, b',
\qreg{D})$ is the same distribution as $\cZ^i(\theta^i, b',
\qreg{D})$. Therefore, we have $\big(\qreg{C}, \cZ^i(\theta^i,
b', \qreg{D})\big) \approx_c \big(\qreg{C}, \cZ^i(0^\secp, b',
\qreg{D})\big)$. This concludes the proof of Lemma \ref{lemma:i}.
\end{proof}

Next, assume for contradiction that there exists a malicious QPT verifier
$\Ver^*$ and an unbounded distinguisher $\sD$ that distinguishes the
output of $\Hyb_2^{i-1}$ and $\Hyb_2^{i}$ with non-negligible
probability. We will now construct a QPT adversary $\qA$ (along with
an unbounded distinguisher) that
participates in experiment $\Exp_{\cZ^i, \qA}$ specified by Theorem
\ref{thm:cert-del} and breaks its certified-deletion security.

$\underline{\Exp_{\cZ^i, \qA}(b):}$
\begin{enumerate}
\item Sample $y^i, \theta^i \la \bit^\secp$ and
initialize $\qA$ with $1^\secp$ along with the following input (or
the input $\bot$ if $\cZ^i$ does output $\bot$):

$$\big(\rho'_{\qreg{R}\otimes\qreg{P}\otimes\qreg{S}\otimes\qreg{B}},
\{\theta^j, y^j\}_{j \in [2\ell] \setminus \{i\}}\big)\seteq
\cZ^i\bigg(\theta^i, b \oplus \bigoplus_{j:\theta^i_j = 0}y^i_j,
\ket{y^i}^{\theta^i}\bigg)$$

\item $\qA$ executes as follows:

\begin{itemize}
\item Measure the $\qreg{R}$ register of
$\rho'_{\qreg{R}\otimes\qreg{P}\otimes\qreg{S}\otimes\qreg{B}}$ in the Hadamard basis
to get outcome $\cert$ and leftover state
$\rho''_{\qreg{R}\otimes\qreg{P}\otimes\qreg{S}\otimes\qreg{B}}$.
\item Parse $\cert = \cert^1 \| \ldots \| \cert^{2\ell}$ where
$\cert^j \in \bit^\secp$ for each $j \in [2\ell]$.
\item If for each $j \in [2\ell] \setminus \{i\}$, it holds that
$\cert^j_k = y^j_k$ for each $k \in [\secp]$ such that $\theta^j_k = 1$, then
output $\cert^i$ along with the state $\rhoB''$ on register
$\qreg{B}$. Otherwise, set $\rho''_{\qreg{B}} \seteq \bot$ and output
$\cert^i$ along with $\rho''_{\qreg{B}}$.
\end{itemize}

\item If $\forall j \in [\secp]$ such that $\theta^i_j = 1$,
$\cert^i_j = y^i_j$, output $\rho''_B$ (and $\bot$
otherwise).
\end{enumerate}

Observe now that the output of $\Hyb_2^{i-1}$ is identically
distributed to the output of $\Exp_{\cZ^i, \qA}(0)$ (conditioned on
$\theta^i \neq 0^\secp$ which occurs with $\negl$ probability).
Likewise, $\Hyb_2^i$ is statistically close to $\Exp_{\cZ^i,
\qA}(1)$, assuming a bit flip occurs from $\Hyb_2^{i-1}$ to $\Hyb_2^i$
(otherwise, the hybrids are identically distributed). Hence, the
unbounded distinguisher $\sD$ can also distinguish between
$\Exp_{\cZ^i, \qA}(0)$ and $\Exp_{\cZ^i, \qA}(1)$ with non-negligible
advantage. However, Theorem \ref{thm:cert-del}
guarantees that $\Exp_{\cZ^i, \qA}(0) \approx_s \Exp_{\cZ^i, \qA}(1)$,
a contradiction. Therefore, $\Hyb_2^{i-1} \approx_s \Hyb_2^i$. This
concludes the proof of Lemma \ref{lemma:main}.
\end{proof}

As a consequence of the above lemma, we have that $\Hyb_1 \approx_s
\Hyb_2^\ell$.

Next, consider the hybrids $\Hyb_2^{\ell+1}, \ldots, \Hyb_2^{2\ell}$,
in which the proof $\pi_\inn$ inside $\ct^0$ will be replaced
bit-by-bit with $0s$. In other words, $\Hyb_2^{2\ell}$ is the
following:

$\underline{\Hyb_2^{2\ell}:}$
\begin{itemize}
\item $\Ver^*$ is initialized with
\textcolor{black}{$\crs_\Sim \| \crs_\out$}, where
\textcolor{black}{$(\crs_\Sim, \pi_\Sim) \la \Sim_\nizk(x)$} and
$\crs_\out \la \Setup_\nizk(1^\secp)$, with $\Sim_\nizk$ being the
QPT simulator of $\sPi_\nizk$.
\item Then, $\Ver^*$ is provided with $\sigma =
(\ket{\psi}_{\qreg{R}\otimes \qreg{P} \otimes \qreg{S}}, \ct^0, \ct^1)$
computed as $(\sigma, \qP) \la \Pro(\textcolor{black}{\crs_\Sim \|
\crs_\out}, \allowbreak x, \wit)$, except $\ct^0 =
\textcolor{red}{0^\ell} \oplus \pad^0 \oplus \key^0$ and $\ct^1 =
\pi_\Sim \oplus \pad^1 \oplus \key^1$.
\item Then, $\Ver^*$ outputs a bi-partite state $\rhoAB$. If
$\DelVrfy(\qP, \rhoA) = \top$, output $\rhoB$ where $\rhoA$ is the
state on register $\qreg{A}$ and $\rhoB$ is the residual state on
register $\qreg{B}$. Output $\bot$ otherwise.
\end{itemize}

\begin{lemma}\label{lemma:final}
$\forall i \in [\ell]: \Hyb_2^{\ell + i - 1} \approx_s \Hyb_2^{\ell+i}$
\end{lemma}
\begin{proof}
This follows from a similar argument as the proof of Lemma
\ref{lemma:main}. This is due to the fact that the statement $x^\sz$ remains
true whether a particular bit inside $\ct^0$ is flipped or not. This is because
$\ct^1$ already consists of $\pi_\Sim$ which is sufficient to satisfy
the second clause of the OR condition of statement $x^\sz$,
due to the computational zero-knowledge property of $\sPi_\nizk$.
Consequently, we can invoke the superposition-secure adaptive zero-knowledge
guarantee for the outer proofs as in the proof of Lemma
\ref{lemma:main}. This finishes the proof of Lemma \ref{lemma:final}.
\end{proof}

Using this lemma, we have $\Hyb_2^\ell \approx_s
\Hyb_2^{2\ell}$, which gives us that $\Hyb_0 \approx_s
\Hyb_{2}^{2\ell}$. Notice that in $\Hyb_2^{2\ell}$, the witness $\wit$
is never used and that the view of $\Ver^*$
can be simulated entirely with the statement $x$. This proves
Theorem \ref{thm:crs}.
\end{proof}


\section{CE-NIZK in the Shared EPR Model}\label{sec:crqs_pos}

\subsection{The Hidden-Bits Paradigm}

\begin{definition}[Hidden Bits Generator \cite{EC:QuaRotWic19}]
\label{def:hbg}
A hidden bits generator (HBG) is a triple of algorithms $(\Setup,
\GenBits, \Verify)$ with the following syntax:
\begin{itemize}
\item $\Setup(1^\secp, 1^k) \ra \crs:$ The setup algorithm takes as
input a security parameter and the number of hidden-bits $k$, and
outputs a common reference string $\crs$.
\item $\GenBits(\crs) \ra \big(\com, r, \{\pi_i\}_{i\in[k]}\big):$ The
hidden bits generation algorithm
outputs a commitment $\com$, hidden bits $r \in \bit^k$ and openings
$\{\pi_i\}_{i\in[k]}$.
\item $\Verify(\crs, \com, i, r_i, \pi_i) \ra \top/\bot:$ The
verification algorithm takes an index $i$, hidden bit $r_i$ and proof
$\pi_i$ along with a CRS $\crs$ and commitment $\com$. It outputs
$\top$ (accept) or $\bot$ (reject).
\end{itemize}

An HBG must satisfy the following properties:
\begin{itemize}
\item \emph{Correctness:} For all $i \in [k]$, the following holds:
\[
\Pr\left[
\Verify(\crs,\com,i,r_i,\pi_i) = \bot
 \ :
\begin{array}{rl}
 &\crs \gets \Setup(1^\secp, 1^k)\\
 & (\com,r,\pi_{[k]}) \la \GenBits(\crs)
\end{array}
\right]
\le \negl
\]

\item \emph{Succinct Commitment:} There exists a set of commitments
$\cC$ such that for all $\crs \la \Setup(1^\secp, 1^k)$ and every $\com$
output by $\GenBits(\crs)$, $\com \in \cC$. Moreover, there exists
a constant $\delta < 1$ such that $|\cC| \le
2^{k^{\delta}\poly(\secp)}$. Furthermore, $\forall \com \notin \cC$,
$\Verify(\crs, \com, \cdot, \cdot, \cdot)$ outputs $\bot$.

\item \emph{Statistical Binding:} There exists a (possibly inefficient)
deterministic algorithm $\Open(1^k, \allowbreak\crs, \com)$ such that 
for all $k = \poly(\secp), \crs$ and $\com$, it outputs $r$ such that
for every every unbounded $\tPro$, the following holds:

\[
\Pr\left[
\begin{array}{rl}
&r_i^* \neq r_i\\
&\land\\
&\Verify(\crs, \com, i, r_i^*, \pi_i) = \top
\end{array}
\ :
\begin{array}{rl}
 &\crs \gets \Setup(1^\secp, 1^k)\\
 & (\com, i, r_i^*, \pi_i) \la \tPro(\crs)\\
 & r \la \Open(1^k,\crs,\com)
\end{array}
\right]
\le \negl
\]

\item \emph{Computational Hiding:} For all $k = \poly(\secp)$ and $I
\subseteq [k]$, the following holds:

\[
\left[
\begin{array}{ll}
    (\crs,\com, &I,r_I,\pi_I,r_{\overline{I}})\\
 &\approx_c\\
    (\crs,\com,&I,r_I,\pi_I,r'_{\overline{I}})
\end{array}
 \ :
\begin{array}{rl}
 &\crs \gets \Setup(1^\secp, 1^k)\\
 & (\com,r,\pi_{[k]}) \la \GenBits(\crs)\\
 & r' \la \bit^k
\end{array}
\right]
\]
\end{itemize}
\end{definition}

\begin{theorem}[\cite{EC:WatWeeWu25}]\label{thm:lwe_hbg} Assuming the polynomial hardness
of LWE, there exists an HBG with statistical binding
and computational hiding in the CRS model.
\end{theorem}

\begin{definition}[NIZK in the Hidden Bits Model \cite{FLS90,FLS99}] A NIZK proof in the
hidden bits model (HBM) for a language $\cL \in \NP$ is a pair of
algorithms $(\Pro, \Ver)$ with the following syntax, where $k(\secp,n) =
\poly(\secp, n)$:

\begin{itemize}
\item $\Pro(r, x, \wit) \ra (I, \pi):$ The prover takes a hidden bit string $r \in
\bit^{k(\secp,n)}$, a statement $x$ with size $|x| = n$, and a witness
$\wit$ as input. It outputs a set of indices $I \subseteq [k]$ and a proof
$\pi$.

\item $\Ver(I, r_{I}, x, \pi) \ra \top/\bot:$ The verifier on input a set of indices
$I$, hidden bits $r_{I} = \{r_i\}_{i\in I}$, statement $x$ and proof
$\pi$, outputs $\top/\bot$.
\end{itemize}

Such a NIZK must satisfy the following properties:

\begin{itemize}
\item \emph{Completeness:} For every $x \in \cL$ of size $|x|=n$ and $\wit \in
\RL(x)$, we have:

\[
\Pr\left[
\Ver\big(I, r_I, x, \pi) \ra 0
 \ :
\begin{array}{rl}
 &r \la \bit^{k(\secp,n)}\\
 & (I, \pi) \la \Pro(r,x,\wit)
\end{array}
\right] \le \negl
\]

\item \emph{Soundness:} For all $n=\poly(\secp)$ and every unbounded
$\Pro^*$, the following holds:

\[
\Pr\left[
\Ver\big(I, r_I, x, \pi) \ra 1 \boldsymbol{\land} |x| = n
\boldsymbol{\land} x \notin \cL
 \ :
\begin{array}{rl}
 &r \la \bit^{k(\secp,n)}\\
 & (x, \pi, I) \la \Pro^*(r)
\end{array}
\right] \le \negl
\]

\item \emph{Zero-Knowledge:} There exists a simulator $\Sim$ such that
for every $x \in \cL$ and $\wit \in \RL(x)$, the following holds:

\[
\left[
(x, I, r_I, \pi) \equiv (x, I', r_I', \pi')
 \ :
\begin{array}{rl}
 &r \la \bit^{k(\secp,n)}\\
 & (I, \pi) \la \Pro(r, x, \wit)\\
 & (I', r'_I, \pi') \la \Sim(x)
\end{array}
\right]
\]
\end{itemize}
\end{definition}

\begin{theorem}[\cite{FLS90,FLS99}]\label{thm:fls} NIZKs exist in
the HBM for every language $\cL \in \NP$.
\end{theorem}

\begin{remark}[Amplifying Soundess \cite{EC:QuaRotWic19}]
For any polynomial $q(\secp, n)$, given a NIZK in the HBM (requiring $k'$ many hidden bits) with $\negl$ soundness error, we can build one with $2^{-q(\secp,
n)}\cdot \negl$ soundness error. This is obtained by a $q$-fold parallel
repetition of the base NIZK, requiring $q\cdot k'$ many hidden bits.
\end{remark}

\subsection{Construction from any Hidden Bits Generator}
\label{sec:crqs_const}

\begin{framed}
\begin{center}
\textbf{Certified-Everlasting NIZK in the Shared EPR Model}
\end{center}
\noindent \underline{Building Blocks:}
\begin{enumerate}
\item Hidden Bits Generator $\HBG \seteq (\Setup_\bg, \GenBits_\bg,
\Verify_\bg)$.

\item NIZK in the Hidden Bits Model $\sPi_\hb \seteq (\Pro_\hb, \Ver_\hb)$.

\item Let $k \seteq \secp$ and $\ell$ be a sufficiently large polynomial in
$\secp$ (the exact size is determined in the proof of
statistical soundness, based on the number of hidden bits $k'$ required by
$\sPi_\hb$).
For each $(i,j)\in [\ell]\times[k]$, consider EPR pairs of the
form $\qreg{P^i_j}, \qreg{V_j^i}$. For each $i \in [\ell]$, set
$\qreg{P^i} \seteq \qreg{P_1^i} \otimes \ldots \otimes \qreg{P_k^i}$
and likewise for $\qreg{V^i}$. Finally, set $\qreg{P} \seteq
\qreg{P^1} \otimes \ldots \otimes \qreg{P^\ell}$ and likewise for
$\qreg{V}$.
\end{enumerate}

\noindent $\underline{\Setup(1^\secp):}$
\begin{itemize}
    \item Sample $\crs_\bg \la \Setup_\bg(1^\secp, 1^{k\ell})$ and 
$s \la \bit^\ell$.
\item Output $\crs \seteq (\crs_\bg, s)$.
\end{itemize}

\noindent $\underline{\Pro(\crs, \qreg{P}, x, \wit):}$
\begin{enumerate}
\item Parse $\crs = (\crs_\bg, s)$ and
$\qreg{P} = \qreg{P^1} \otimes \ldots \otimes \qreg{P^\ell}$.
\item Compute $(\com, \theta, \{\op^{(i,j)}\}_{(i,j)\in
[\ell]\times[k]}) \la \GenBits_\bg(\crs_\bg)$, where
$\theta = \{\theta^i\}_{i\in[\ell]}$ and $\theta^i \in \bit^k$
for each $i\in[\ell]$, while $\op^{(i,j)}$ denotes the opening for bit
$\theta^i_j$.
\item For each $i \in [\ell]$, execute the following:
\begin{enumerate}
\item Measure $\qreg{P^i}$ in the $\theta^i$ basis to get outcome $y^i$.
\item Compute $t_i \seteq \bigoplus\limits_{j: \theta^i[j] = 0}y_j^i$.
\end{enumerate}
\item Compute $r_i \seteq t_i \oplus s_i$ for each $i \in [\ell]$.

\item Execute $(I \subseteq [\ell], \pi_\hb) \la \Pro_\hb(r, x, \wit)$.
\item Output $\pi \seteq (I, \pi_\hb, \com, \theta^I = \{\theta^i\}_{i\in I},
    \{\op^{(i,j)}\}_{(i,j)\in I \times [k]})$ and $\qP \seteq (y \seteq y^1 \| \ldots \| y^\ell,
\theta, I, \qreg{P})$.
\end{enumerate}

\noindent $\underline{\Ver(\crs, \qreg{V}, x, \pi):}$
\begin{enumerate}
\item Parse $\pi = (I, \pi_\hb, \com, \theta^I = \{\theta^i\}_{i\in I},
\{\op^{(i,j)}\}_{i\in I, j \in [k]})$, $\crs = (\crs_\bg, s)$ and $\qreg{V}
= \qreg{V^1} \otimes \ldots \otimes \qreg{V^\ell}$.

\item For each $i \in I$, do the following:
\begin{itemize}
\item For each $j \in [k]$, if $\Verify_\bg(\crs_\bg, \com, (i, j),
\theta^i_j, \op^{(i,j)}) \neq
\top$, output $\bot$.
\item Measure $\qreg{V^i}$ in the $\theta^i$ basis to get $y^i$.
\item Compute $t_i \seteq \bigoplus\limits_{j: \theta^i[j] = 0}y_j^i$.
\end{itemize}

\item Compute $r_I \seteq t_I \oplus s_I$.

\item Output the bit $b \seteq \Ver_\hb(I, r_I, x, \pi_{\hb})$ along with
    $\qV \seteq (I, \qreg{V})$.
\end{enumerate}

\noindent $\underline{\Del(\qV):}$
\begin{enumerate}
\item Parse $\qV = (I, \qreg{V})$ and $\qreg{V} = \qreg{V^1}
\otimes \ldots \otimes \qreg{V^\ell}$.
\item For each $i \in [\ell] \setminus I$, measure $\qreg{V^i}$ in the
Hadamard basis to get $\cert^i$.
\item Output $\cert \seteq \{\cert^i\}_{i\in{[\ell] \setminus I}}$.
\end{enumerate}

\noindent $\underline{\DelVrfy(\cert, \qP):}$
\begin{enumerate}
\item Parse $\cert = \{\cert^i\}_{i\in[\ell] \setminus I}$,
$\qP = (y^1 \| \ldots \| y^\ell, \theta, I, \qreg{P})$.

\item If there exists $i \in [\ell] \setminus I$ and $j \in [k]$ such that
$\theta^i_j = 1$ and $\cert^i_j \neq y^i_j$, output $\bot$. Otherwise,
output $\top$.
\end{enumerate}
\end{framed}

\begin{theorem}
Let $\HBG$ be any statistical-binding hidden bits generator and
$\sPi_\hb$ be a NIZK for $\NP$ in the hidden bits model. Then, $\sPi =
(\Setup, \Pro, \Ver, \Del, \DelVrfy)$ is a CE-NIZK for $\NP$ in the shared EPR
model.
\end{theorem}

Hence, from Theorem \ref{thm:lwe_hbg} and Theorem \ref{thm:fls}, we have the following:

\begin{corollary}
Assuming the polynomial hardness of LWE, there exists a CE-NIZK for
$\NP$ in the shared EPR model.
\end{corollary}

\begin{proof}
We show that $\sPi$ satisfies the properties of CE-NIZK as follows:

\subsubsection{Completeness and Deletion Correctness}

\paragraph{Completeness.} The hidden-bits NIZK proof $\pi_\hb$ is
computed as $(I, \pi_\hb) \la \Pro_\hb(r, x, \wit)$ where $r = t \oplus s$.
Clearly, the distribution of $r$ is uniform. By correctness of
the HBG $\HBG$, all of the $\Verify_\bg(\crs_\bg, \com, (i,j),
\theta^i_j, \op^{(i, j)})$ checks will pass. Moreover, the verifier $\Ver$
computes the correct $r_I$ values as $s_I$ is available as part of
$\crs$ and the $t_I$ values are common due to the EPR pairs being
measured in the same basis on both ends for the indices corresponding
to $I$.
Finally, $\Ver_\hb(I, r_I, x, \pi_\hb)$ outputs $1$ by the
completeness of the hidden-bits NIZK $\sPi_\hb$. Consequently, $\Ver$
outputs $1$.

\paragraph{Deletion Correctness.}
This follows directly from the fact that for each $i \in [\ell] \setminus I$ and
each $j \in [k]$ such that $\theta^i_j = 1$, both parties measure the EPR halves 
in the Hadamard basis to get $\cert^i_j = y^i_j$.

\subsubsection{Computational Zero-Knowledge}\label{sec:epr-czk}

For every $x\in L$ and $\omega\in R_L(x)$, consider the following hybrids.

\medskip
\noindent\underline{$\Hyb^{\mathsf{zk}}_0$}: This is the real experiment, which
proceeds as follows:
\begin{enumerate}
\item $\crs_\bg \leftarrow \Setup_\bg(1^\lambda,1^{k\ell})$ and
$s\leftarrow\bit^\ell$.
\item $(\com,\theta,\{\op^{(i,j)}\}_{(i,j)\in[\ell]\times[k]})
\leftarrow\GenBits_{\mathsf{bg}}(\crs_{\mathsf{bg}})$.
\item For each $i\in[\ell]$, measure $\qreg{P^i}$ in basis $\theta^i$ to obtain
$y^i$, and compute
\[
 t_i:=\bigoplus_{j:\theta^i_j=0}y^i_j.
\]
Let $t=t_1\|\cdots\|t_\ell$.
\item Compute $r:=t\oplus s$ and
$(I,\pi_{\mathsf{hb}})\leftarrow\Pro_\hb(r,x,\omega)$.
\item Output
\[
 \crs:=(\crs_{\mathsf{bg}},s),\qquad
 \pi:=(I,\pi_{\mathsf{hb}},\com,\{\theta^i\}_{i\in I},
       \{\op^{(i,j)}\}_{(i,j)\in I\times[k]})
\]
and the verifier register $\qreg{V}$.
\end{enumerate}

\noindent\underline{$\Hyb^{\mathsf{zk}}_1$}: Here, $r$ and the opening set
$I$ are determined before the HBG basis string is generated.
\begin{enumerate}
\item $\crs_{\mathsf{bg}}\leftarrow
\Setup_{\mathsf{bg}}(1^\lambda,1^{k\ell})$ and
$r\leftarrow\bit^\ell$.
\item Compute $(I,\pi_{\mathsf{hb}})\leftarrow\Pro_\hb(r,x,\omega)$.
\item $(\com,\theta,\{\op^{(i,j)}\}_{(i,j)\in[\ell]\times[k]})
\leftarrow\GenBits_{\mathsf{bg}}(\crs_{\mathsf{bg}})$.
\item For each $i\in[\ell]$, measure $\qreg{P^i}$ in basis $\theta^i$ to obtain
$y^i$ and compute $t_i:=\bigoplus_{j:\theta^i_j=0}y^i_j$.
Let $t=t_1\|\cdots\|t_\ell$.
\item Compute $s:=t\oplus r$ and output $(\crs,\pi,\qreg{V})$ as in
$\Hyb^{\mathsf{zk}}_0$.
\end{enumerate}

\noindent\underline{$\Hyb^{\mathsf{zk}}_2$}: This is the same as
$\Hyb^{\mathsf{zk}}_1$, except that the unopened EPR blocks are measured in
fresh random bases.
\begin{enumerate}
\item $\crs_{\mathsf{bg}}\leftarrow
\Setup_{\mathsf{bg}}(1^\lambda,1^{k\ell})$ and
$r\leftarrow\bit^\ell$.
\item Compute $(I,\pi_{\mathsf{hb}})\leftarrow\Pro_\hb(r,x,\omega)$.
\item $(\com,\theta,\{\op^{(i,j)}\}_{(i,j)\in[\ell]\times[k]})
\leftarrow\GenBits_{\mathsf{bg}}(\crs_{\mathsf{bg}})$.
\item For every $i\in[\ell]\setminus I$, sample
$\widetilde\theta^i\leftarrow\bit^k$ independently.  For every
$i\in[\ell]$, measure $\qreg{P^i}$ in basis $\theta^i$ if $i\in I$, and in basis
$\widetilde\theta^i$ if $i\notin I$.  Denote the outcome by $y^i$ and set
\[
 t_i:=
 \begin{cases}
 \displaystyle\bigoplus_{j:\theta^i_j=0}y^i_j,&i\in I,\\[2mm]
 \displaystyle\bigoplus_{j:\widetilde\theta^i_j=0}y^i_j,&i\notin I.
 \end{cases}
\]
Let $t=t_1\|\cdots\|t_\ell$.
\item Compute $s:=t\oplus r$ and output $(\crs,\pi,\qreg{V})$ as in
$\Hyb^{\mathsf{zk}}_0$.
\end{enumerate}

\noindent\underline{$\Hyb^{\mathsf{zk}}_3$}: This is the same as
$\Hyb^{\mathsf{zk}}_2$, except that the CRS bits corresponding to unopened
indices are sampled uniformly.  Namely, after computing $t$, set
\[
 s_i:=
 \begin{cases}
 t_i\oplus r_i,&i\in I,\\
 u_i,&i\notin I,
 \end{cases}
 \qquad u_i\leftarrow\bit,
\]
where the values $\{u_i\}_{i\notin I}$ are mutually independent. The output
is $(\crs,\pi,\qreg{V})$ as before.\\

\noindent\underline{$\Hyb^{\mathsf{zk}}_4$}: Finally, the tuple
$(I,r_I,\pi_{\mathsf{hb}})$ is generated using the hidden-bits simulator.
\begin{enumerate}
\item $\crs_{\mathsf{bg}}\leftarrow
\Setup_{\mathsf{bg}}(1^\lambda,1^{k\ell})$ and
$(I,r_I,\pi_{\mathsf{hb}})\leftarrow\Sim_\hb(x)$.
\item $(\com,\theta,\{\op^{(i,j)}\}_{(i,j)\in[\ell]\times[k]})
\leftarrow\GenBits_{\mathsf{bg}}(\crs_{\mathsf{bg}})$.
\item For every $i\in I$, prepare $k$ EPR pairs, measure the prover halves
$\qreg{P^i} \seteq \qreg{P^i_1} \otimes \ldots \otimes \qreg{P^i_k}$ in basis $\theta^i$ to obtain $y^i$, compute
$t_i:=\bigoplus_{j:\theta^i_j=0}y^i_j$, and set
$s_i:=t_i\oplus r_i$.
\item For every $i\notin I$, sample $s_i\leftarrow\bit$ and prepare the
register $\qreg{V^i}$ in the maximally mixed state $\mathbb I/2^k$.
\item Output
\[
 \crs:=(\crs_{\mathsf{bg}},s),\qquad
 \pi:=(I,\pi_{\mathsf{hb}},\com,\{\theta^i\}_{i\in I},
       \{\op^{(i,j)}\}_{(i,j)\in I\times[k]})
\]
and the verifier register $\qreg{V}:=\qreg{V^1}\otimes\cdots\otimes \qreg{V^\ell}$.
\end{enumerate}

We now prove the indistinguishability of the hybrids.

\begin{lemma}
$\Hyb^{\mathsf{zk}}_0\equiv\Hyb^{\mathsf{zk}}_1$.
\end{lemma}
\begin{proof}
Sampling $s$ uniformly and setting
$r \seteq t\oplus s$ is identical to sampling $r$ uniformly and setting
$s \seteq t\oplus r$. Hence, the claim follows.
\end{proof}

\begin{lemma}
$\Hyb^{\mathsf{zk}}_1\approx_c\Hyb^{\mathsf{zk}}_2$.
\end{lemma}
\begin{proof}
Observe that $I$ is determined before the hidden-bits are generated, so applying
computational hiding of the HBG with revealed set $I\times[k]$ gives us:
\begin{align*}
&(\crs_{\mathsf{bg}},\com,I,\theta_I,\{\op^{(i, j)}\}_{(i,j)\in I \times [k]},
\{\theta^i\}_{i \in [\ell]\setminus I})\\
&\;\;\;\;\;\;\;\;\;\;\;\;\;\;\;\;\;\;\;\;\;\;\;\;\;\;\;\;\;\;\;\;\approx_c\\
&(\crs_{\mathsf{bg}},\com,I,\theta_I,\{\op^{(i, j)}\}_{(i,j)\in I \times [k]},
\{\widetilde\theta^i\}_{i \in [\ell]\setminus I}),
\end{align*}
where the values $\{\widetilde\theta^i\}_{i \in [\ell]\setminus I}$ are uniformly
random and independent of other values. The remainder of either hybrid is 
efficient post-processing of this tuple. The last components of the
tuples are used as
measurement bases for the ``unopened'' EPR blocks. Then, the measurement
outcomes are used to compute $t, s$ and the
resulting view. Hence, these hybrids are computationally indistinguishable.
\end{proof}

We first establish the following lemma, before showing the
indistinguishability of $\Hyb^\zk_2$ and $\Hyb^\zk_3$.

\begin{lemma}\label{lem:epr-parity}
Let
\[
\lvert\Phi_k\rangle_{\qreg{P}\otimes\qreg{V}}:=
2^{-k/2}\sum_{z\in\bit^k}\lvert z\rangle_{\qreg{P}}\lvert z\rangle_{\qreg{V}}.
\]
Sample $\Theta\leftarrow\bit^k$, measure $\qreg{P}$ in the basis $\Theta$
to obtain $Y\in\bit^k$, and write
\[
 T:=\bigoplus_{j:\Theta_j=0}Y_j,
\]
on a 1-qubit register $\qreg{T}$. Note that the XOR over the empty set is defined
to be $0$. Let $\rho_{\qreg{T}\otimes\qreg{V}}$
denote the resulting state on registers $\qreg{T}, \qreg{V}$. Then,
if $\Theta$ is not included in the output, the following holds:
\[
\TD\!\left(\rho_{\qreg{T}\otimes\qreg{V}},
\frac{\mathbb I_\qreg{T}}{2}\otimes\frac{\mathbb I_\qreg{V}}{2^k}\right)
=2^{-k-1}.
\]
The same bound holds after adjoining arbitrary side information independent
of $(\Theta,\qreg{P},\qreg{V})$ and after replacing $T$ by $T\oplus R$, where $R$
is a classical bit computed from that side information.
\end{lemma}
\begin{proof}
For $\theta\in\bit^k$, write
$H^\theta:=\bigotimes_{j=1}^kH^{\theta_j}$ and
$p_\theta(y):=\bigoplus_{j:\theta_j=0}y_j$. Conditioned on
$(\Theta,Y)=(\theta,y)$, register $\qreg{V}$ is in state $H^\theta\lvert y\rangle$.
For $b\in\bit$, let
\[
 \sigma_b:=4^{-k}
 \sum_{\substack{\theta,y\in\bit^k:\ p_\theta(y)=b}}
 H^\theta\ketbra{y}{y}H^\theta.
\]
Thus,
\[
    \rho_{\qreg{T}\otimes\qreg{V}}=\ketbra{0}{0}\otimes\sigma_0+
    \ketbra{1}{1}\otimes\sigma_1 \;\;\;\;\;\text{and}\;\;\;\;\;
    \sigma_0+\sigma_1=\frac{\mathbb I_\qreg{V}}{2^k}.
\]

Moreover,
\begin{align*}
\sigma_0-\sigma_1
&=4^{-k}\sum_{\theta,y}(-1)^{p_\theta(y)}
H^\theta\ketbra{y}{y}H^\theta\\
&=
4^{-k}\bigotimes_{j=1}^k
\left(
\sum_{\theta_j\in\{0,1\}}
\sum_{y_j\in\{0,1\}}
(-1)^{(1-\theta_j)y_j}
H^{\theta_j}\ketbra{y_j}{y_j}H^{\theta_j}
\right)\\
&=4^{-k}(Z+\mathbb I)^{\otimes k}
=2^{-k}\ketbra{0^k}{0^k}.
\end{align*}
The second equality follows coordinate-wise. If $\theta_j=0$, the
contribution is $\sum_{y_j}(-1)^{y_j}\ketbra{y_j}{y_j}=Z$. If
$\theta_j=1$, the contribution is
$\sum_{y_j}H\ketbra{y_j}{y_j}H=\mathbb I$.
Since the two diagonal blocks of
$\rho_{\qreg{T}\otimes\qreg{V}}-(\mathbb I_{\qreg{T}}/2)\otimes(\mathbb I_{\qreg{V}}/2^k)$ are
$(\sigma_0-\sigma_1)/2$ and its negative, respectively,
\[
\TD\!\left(\rho_{\qreg{T}\otimes\qreg{V}},
\frac{\mathbb I_\qreg{T}}{2}\otimes\frac{\mathbb I_\qreg{V}}{2^k}\right)
 =\frac12\lVert\sigma_0-\sigma_1\rVert_1
 =2^{-k-1}.
\]
Also, tensoring an independent state and applying a controlled XOR are
trace-distance preserving, while discarding registers cannot increase trace
distance. This proves the final statement.
\end{proof}

\begin{lemma}
$\Hyb^{\mathsf{zk}}_2\approx_s\Hyb^{\mathsf{zk}}_3$.
\end{lemma}
\begin{proof}
Enumerate the indices in $[\ell]\setminus I$ and replace the corresponding
bits of $s$ one at a time.  At a target index $i\notin I$, condition on all
classical and quantum side information outside block $i$, including $r_i$
and the hidden-bits proof. The basis
$\widetilde\theta^i$ and the EPR block $(\qreg{P^i},\qreg{V^i})$ are independent of this
side information. Lemma~\ref{lem:epr-parity}, followed by the map
$s_i=t_i\oplus r_i$, bounds the trace distance of this switch by $2^{-k-1}$. Therefore,
\[
 \TD(\Hyb^{\mathsf{zk}}_2,\Hyb^{\mathsf{zk}}_3)
 \leq \ell\,2^{-k-1}=\negl
\]
since $\ell = \poly(\secp)$ and $k=\lambda$.
\end{proof}

\begin{lemma}
$\Hyb^{\mathsf{zk}}_3\approx_s\Hyb^{\mathsf{zk}}_4$.
\end{lemma}
\begin{proof}
In $\Hyb^{\mathsf{zk}}_3$, the unopened values $s_i$ are uniform and no
longer depend on $r_i$ or $t_i$.  After the hidden random measurement basis
and outcome of an unopened EPR block are discarded, the verifier's marginal
state is exactly $\mathbb I/2^k$.  Consequently, the remainder of the hybrid
depends on the hidden-bits execution only via
$(I,r_I,\pi_{\mathsf{hb}}) \la \Pro_\hb(r, x, \wit)$. By perfect zero-knowledge of the NIZK in the
hidden-bits model, we have
\[
 (x,I,r_I,\pi_{\mathsf{hb}})
 \equiv
 (x,I',r'_{I'},\pi'_{\mathsf{hb}}),
 \qquad
 (I',r'_{I'},\pi'_{\mathsf{hb}})\leftarrow\Sim_\hb(x).
\]
Hence, the indistinguishability follows.
\end{proof}

Observe that $\Hyb^\zk_4$ can be computed using the statement $x$ alone. This
finishes the proof of computational zero-knowledge.

\subsubsection{Certified-Everlasting Zero-Knowledge}\label{sec:epr-cezk}

Fix a QPT malicious verifier $\Ver^*$, $x\in L$, and $\omega\in R_L(x)$. Consider
the following hybrids.

\medskip
\noindent\underline{$\Hyb^{\mathsf{cezk}}_0$}: This is the real experiment,
which proceeds as follows:
\begin{enumerate}
\item $\crs_{\mathsf{bg}}\leftarrow
\Setup_{\mathsf{bg}}(1^\lambda,1^{k\ell})$ and
$s\leftarrow\bit^\ell$.
\item $(\com,\theta,\{\op^{(i,j)}\}_{(i,j)\in[\ell]\times[k]})
\leftarrow\GenBits_{\mathsf{bg}}(\crs_{\mathsf{bg}})$.
\item For every $i\in[\ell]$, measure $\qreg{P^i}$ in basis $\theta^i$ to obtain
$y^i$, and compute
$t_i:=\bigoplus_{j:\theta^i_j=0}y^i_j$.  Let
$t=t_1\|\cdots\|t_\ell$.
\item Compute $r:=t\oplus s$ and execute
$(I,\pi_{\mathsf{hb}})\leftarrow\Pro_\hb(r,x,\omega)$.
\item Set
\[
 \crs:=(\crs_{\mathsf{bg}},s),\qquad
 \pi:=(I,\pi_{\mathsf{hb}},\com,\theta_I,
       \{\op^{(i,j)}\}_{(i,j)\in I\times[k]}).
\]
\item Run $(\cert^*,\rho^*)\leftarrow \Ver^*(\crs,\qreg{V},x,\pi)$.
\item Let $\rho_\Pro \seteq (y \seteq y^1 \| \ldots \| y^\ell, \theta, I,
\qreg{P})$.
\item Output $\rho^*$ if $\DelVrfy(\cert^*, \rho_\Pro) = \top$ and $\bot$ otherwise.
\end{enumerate}

\noindent\underline{$\Hyb^{\mathsf{cezk}}_1$}: Here, $r$ and the opening set
$I$ are determined before the HBG basis is generated.
\begin{enumerate}
\item $\crs_{\mathsf{bg}}\leftarrow
\Setup_{\mathsf{bg}}(1^\lambda,1^{k\ell})$ and
$r\leftarrow\bit^\ell$.
\item Execute $(I,\pi_{\mathsf{hb}})\leftarrow\Pro_\hb(r,x,\omega)$.
\item 
$
(\com,\theta,\{\op^{(i,j)}\}_{(i,j)\in[\ell]\times[k]})
\leftarrow\GenBits_{\mathsf{bg}}(\crs_{\mathsf{bg}}).
$
\item For every $i\in[\ell]$, measure $\qreg{P^i}$ in basis $\theta^i$ to obtain
$y^i$, and compute $t_i:=\bigoplus_{j:\theta^i_j=0}y^i_j$.
Let $t=t_1\|\cdots\|t_\ell$.
\item Set $s:=t\oplus r$, and compute $(\crs,\pi)$ similar to
$\Hyb^{\mathsf{cezk}}_0$.
\item Run $(\cert^*,\rho^*)\la\Ver^*(\crs,\qreg{V},x,\pi)$.
\item Output $\rho^*$ if $\DelVrfy(\cert^*, \rho_\Pro) = \top$ and $\bot$ otherwise.
\end{enumerate}

For each execution, write
$[\ell]\setminus I=\{j_1<\cdots<j_m\}$.  For
$q\in\{0,\ldots,\ell\}$, let
$J_q:=\{j_1,\ldots,j_{\min\{q,m\}}\}$.

\medskip
\noindent\underline{$\Hyb^{\mathsf{cezk},q}_2$}: This is the same as
$\Hyb^{\mathsf{cezk}}_1$, except that after $t$ is computed, the CRS mask is
set as
\[
 s_i:=
 \begin{cases}
 U_i, & i\in J_q,\\
 t_i\oplus r_i, & i\notin J_q,
 \end{cases}
\]
where the $U_i$ are independent uniform bits.  Thus,
$\Hyb^{\mathsf{cezk},0}_2\equiv\Hyb^{\mathsf{cezk}}_1$, while in
$\Hyb^{\mathsf{cezk},\ell}_2$ every $s_i$ with $i\notin I$ is uniform and
independent.

\medskip
\noindent\underline{$\Hyb^{\mathsf{cezk}}_3$}: Finally, the hidden-bits
transcript is generated using the hidden-bits simulator.
\begin{enumerate}
\item $\crs_{\mathsf{bg}}\leftarrow
\Setup_{\mathsf{bg}}(1^\lambda,1^{k\ell})$ and
$(I,r_I,\pi_{\mathsf{hb}})\leftarrow\Sim_\hb(x)$.
\item Compute
\[
(\com,\theta,\{\op^{(i,j)}\}_{(i,j)\in[\ell]\times[k]})
\leftarrow\GenBits_{\mathsf{bg}}(\crs_{\mathsf{bg}}).
\]
\item For every $i\in[\ell]$, measure $\qreg{P^i}$ in basis $\theta^i$ to obtain
$y^i$, and compute $t_i:=\bigoplus_{j:\theta^i_j=0}y^i_j$.
\item Set $s_I:=t_I\oplus r_I$ and sample
$s_i\leftarrow\bit$ independently for every $i\notin I$.
\item Set
\[
 \crs \seteq (\crs_{\mathsf{bg}},s),\qquad
 \pi \seteq (I,\pi_{\mathsf{hb}},\com,\theta_I,
       \{\op^{(i,j)}\}_{(i,j)\in I\times[k]}).
\]
\item Run $(\cert^*,\rho^*)\leftarrow \Ver^*(\crs,\qreg{V},x,\pi)$.
\item Let $\rho_\Pro \seteq (y \seteq y^1 \| \ldots \| y^\ell, \theta, I,
\qreg{P})$.
\item Output $\rho^*$ if $\DelVrfy(\cert^*, \rho_\Pro) = \top$ and $\bot$ otherwise.
\end{enumerate}

We next describe the two auxiliary experiments used to compare adjacent
$\Hyb^{\mathsf{cezk},q}_2$ hybrids.  Fix $q\in[\ell]$ with $q\leq m$, and
let $a:=j_q$.

\medskip
\noindent\underline{$\Exp_q(b)$}: This experiment is the
same as $\Hyb^{\mathsf{cezk},q-1}_2$, except that the target coordinate is set
to
\[
 s_a:=t_a\oplus r_a\oplus b.
\]
The experiment then runs $\Ver^*$ and applies certification as usual, and outputs
$\rho^*$ or $\bot$ accordingly.

\medskip
\noindent\underline{$\tExp_q(b)$} This is the deferred
measurement version of $\Exp_q(b)$.
\begin{enumerate}
\item Generate $\crs_{\mathsf{bg}}$, $r$, $(I,\pi_{\mathsf{hb}})$, and the
HBG output as in $\Hyb^{\mathsf{cezk}}_1$. For every $i \neq a$, measure register
$\qreg{P^i}$, as in that hybrid, but leave $\qreg{P^a}$ unmeasured.
\item Sample $u\leftarrow\bit$ and set $s_a:=u$.  Set all other coordinates
of $s$ as in $\Exp_q(b)$, compute $(\crs,\pi)$, and run
$(\cert^*,\rho^*)\leftarrow \Ver^*(\crs,\qreg{V},x,\pi)$.
\item Measure $\qreg{P^a}$ in basis $\theta^a$ to obtain $y^a$, and compute
$t_a:=\bigoplus_{j:\theta^a_j=0}y^a_j$.
\item Output $\rho^*$ if certification passes and
$u=t_a\oplus r_a\oplus b$.  Otherwise, output $\bot$.
\end{enumerate}

We now prove indistinguishability of the hybrids.

\begin{lemma}
$\Hyb^{\mathsf{cezk}}_0\equiv\Hyb^{\mathsf{cezk}}_1$.
\end{lemma}
\begin{proof}
For every fixed $t\in\bit^\ell$, the map $s\mapsto r=t\oplus s$ is a
bijection on $\bits^\ell$.  Thus, one can first sample
$r\leftarrow\bits^\ell$ and later set $s=t\oplus r$.  The HBG generation and
EPR measurements are independent of $r$ and may be moved after the computation of
$\Pro_\hb(r,x,\omega)$.
\end{proof}

\begin{lemma}\label{lma:iden}
For every $q\leq m$ and $b\in\bit$,
\[
 \Hyb^{\mathsf{cezk},q-1}_2\equiv\Exp_q(0),\qquad
 \Hyb^{\mathsf{cezk},q}_2
 \equiv\tfrac12\Exp_q(0)+\tfrac12\Exp_q(1),
\]
\begin{center}and\end{center}
\[
 \tExp_q(b)\equiv\tfrac12\Exp_q(b)+\tfrac12\bot.
\]
\end{lemma}
\begin{proof}
The first identity is immediate from the definition. In
$\Hyb^{\mathsf{cezk},q}_2$, the target value $s_a$ is uniform, which is equivalent
to sampling $b\leftarrow\bit$ and setting
$s_a=t_a\oplus r_a\oplus b$, giving the second identity. For the last
identity, the measurement of $\qreg{P^a}$ commutes with the action of $\Ver^*$, as
they act on different registers. Hence, this measurement can be done before
running $\Ver^*$. The consistency equality in $\tExp_q(b)$ then holds with probability exactly $1/2$, and conditioned on holding, the resulting experiment is $\Exp_q(b)$.
\end{proof}

Now, fix $q\leq m$ and let $a=j_q$.  In the common state of
$\tExp_q(0)$ and $\tExp_q(1)$ immediately after the output of $\Ver^*$, let
$d:=\cert^{*,a}$ denote the certificate for the corresponding EPR block. For a candidate basis $\vartheta\in\bit^k$, write
$\vartheta_0:=\{j:\vartheta_j=0\}$ and
$\vartheta_1:=\{j:\vartheta_j=1\}$. Define the following projector on $\qreg{P^a}$:
\begin{align*}
 \Pi_{d,\vartheta}:={}&
 \left(H^{\otimes|\vartheta_1|}
 \ketbra{d_{\vartheta_1}}{d_{\vartheta_1}}H^{\otimes|\vartheta_1|}\right)^{\qreg{P^a_{\vartheta_1}}}
 \\
 &\otimes
 \sum_{\substack{z\in\bit^{|\vartheta_0|}:\\
       \Delta(z,d_{\vartheta_0})\geq 1/2}}
 \left(H^{\otimes|\vartheta_0|}
 \ketbra{z}{z}H^{\otimes|\vartheta_0|}\right)^{\qreg{P^a_{\vartheta_0}}},
\end{align*}
where $\Delta$ denotes relative Hamming distance. To include the event that
there are too few computational-basis positions, define
\[
 \widehat\Pi_{d,\vartheta}:=
 \begin{cases}
     \mathbb{I}, & |\vartheta_0|<k/3,\\
 \Pi_{d,\vartheta}, & |\vartheta_0|\geq k/3.
 \end{cases}
\]

\begin{lemma}
In the common pre-measurement state of $\tExp_q(0)$ and $\tExp_q(1)$,
\[
 \Pr[\widehat\Pi_{d,\theta^a}\text{ accepts}]=\negl
\]
\end{lemma}
\begin{proof}
First, replace the candidate basis used in the projector by an independent
$\widetilde\theta^a\leftarrow\bit^k$. It is independent of the state held
by $\Ver^*$, its output $d$, and the unmeasured register $\qreg{P^a}$.  The projector is
diagonal in the Hadamard basis.  Hence, for every fixed outcome
$z$ and fixed $d$, the probability over $\widetilde\theta^a$ that
\[
 z_{\widetilde\theta^a_1}=d_{\widetilde\theta^a_1}
 \quad\text{and}\quad
 \Delta(z_{\widetilde\theta^a_0},d_{\widetilde\theta^a_0})\geq 1/2
\]
is $\negl$ by the same Hoeffding argument as Claim~3.3 in the work of
Bartusek and Khurana \cite{C:BarKhu23}. Also,
$\Pr[|\widetilde\theta^a_0|<k/3]=\negl$. Thus, the bad projector
accepts with negligible probability for an independent candidate basis.

We transfer this bound to the actual HBG basis.  Since $a\notin I$ and $I$
was fixed before HBG generation, apply HBG hiding with all coordinates outside
block $a$ revealed.  A reduction receives the HBG commitment, all outside
basis bits and openings, and either the actual target block $\theta^a$ or an
independent uniform block.  It uses the outside information to generate the
non-target part of the experiment, leaves $\qreg{P^a}$ unmeasured, runs $\Ver^*$, and
measures $\widehat\Pi_{d,\vartheta}$ using the challenge block $\vartheta$.
This measurement is efficient.  A non-negligible difference between the two
acceptance probabilities would therefore violate computational hiding of the
HBG. This argument is similar to the one in Claim 3.4 of \cite{C:BarKhu23}.
\end{proof}

\begin{lemma}\label{lma:xor}
For every $q\leq m$,
$\tExp_q(0)\approx_s\tExp_q(1)$.
\end{lemma}
\begin{proof}
By the previous lemma and gentle measurement, the common global state after
$\Ver^*$ outputs is negligibly close to a state in the image of
$\mathbb{I}-\widehat\Pi_{d,\theta^a}$. Consider the final measurements and tests on
this ideal state. If the target check fails, both experiments output
$\bot$. Suppose it passes. Let $C:=\theta^a_0$ and $n:=|C|\geq k/3$. A purification of the
remaining state has the form
\[
 |\psi\rangle
 =\sum_{\substack{z\in\{0,1\}^n:\ \operatorname{HD}(z,d_C)<n/2}}
 H^{\otimes n}|z\rangle_{P^a_C}|\phi_z\rangle_\qreg{E},
\]
where $\operatorname{HD}$ denotes Hamming distance and $\qreg{E}$ contains all side information and the vectors $|\phi_z\rangle$
may be unnormalized.

Consider applying $Z^{d_C}:=\bigotimes_{j\in C}Z^{d_j}$ for the sake of analysis. Since
\[
 Z^{d_C}H^{\otimes n}|z\rangle
   =H^{\otimes n}|z\oplus d_C\rangle,
 \qquad
 \operatorname{wt}(z\oplus d_C)<n/2,
\]
where $\operatorname{wt}$ denotes Hamming weight. The transformed register has
low-weight support in the Hadamard basis. Hence, the XOR extractor of
\cite{C:BarKhu23} implies that its computational-basis parity
$t_a=\bigoplus_{j\in C}y^a_j$ is uniform and independent of $\qreg{E}$.
The same holds for the original register: $Z^{d_C}$ is diagonal in the
computational basis, so it does not affect the computational basis parity. This
proves $\tExp_q(0)\approx_s\tExp_q(1)$.
\end{proof}

\begin{lemma}
For every $q\in[\ell]$,
$\Hyb^{\mathsf{cezk},q-1}_2\approx_s
\Hyb^{\mathsf{cezk},q}_2$.
\end{lemma}
\begin{proof}
If $q>m$, the hybrids are identical. Otherwise, Lemmas \ref{lma:iden} and
\ref{lma:xor} show that $\Exp_q(0)\approx_s\Exp_q(1)$.  Since
$\Hyb^{\mathsf{cezk},q-1}_2\equiv\Exp_q(0)$ and
$\Hyb^{\mathsf{cezk},q}_2$ is the equal mixture of $\Exp_q(0)$ and
$\Exp_q(1)$, the claim follows. By a hybrid argument, we get:
\[
 \Hyb^{\mathsf{cezk}}_1
 \equiv\Hyb^{\mathsf{cezk},0}_2
 \approx_s\Hyb^{\mathsf{cezk},\ell}_2.
\]
\end{proof}

\begin{lemma}
$\Hyb^{\mathsf{cezk},\ell}_2\approx_s
\Hyb^{\mathsf{cezk}}_3$.
\end{lemma}
\begin{proof}
In $\Hyb^{\mathsf{cezk},\ell}_2$, every $s_i$ for $i\notin I$ is uniform and
independent. Hence, we can replace $(I,r_I,\pi_{\mathsf{hb}})$ by
$(I',r_I',\pi_\hb') \la \Sim_\hb(x)$ using the perfect zero-knowledge of the 
hidden bits NIZK. This yields the hybrid $\Hyb_3^\cezk$.
\end{proof}

Combining the lemmas gives
$\Hyb^{\mathsf{cezk}}_0\approx_s\Hyb^{\mathsf{cezk}}_3$.
The last hybrid is generated by a QPT simulator using only $x$ and black-box
access to $\Ver^*$. Hence, certified-everlasting zero-knowledge follows.

\subsubsection{Statistical Soundness} Let $\Pro^*$ be an unbounded
malicious prover that breaks soundness with some $\mu(\secp)$
probability:

\[
\Pr\left[
\begin{array}{ll}
\Ver(\crs, \qreg{V}, x, \pi^*) = \top\\
\land\\
x \notin \cL
\end{array}
 \ :
\begin{array}{rl}
 &\crs \gets \Setup(1^\secp)\\
 & (x, \pi^*)\la \Pro^*(\crs, \qreg{P})
\end{array}
\right]
\ge \mu(\secp)
\]
where $\pi^* = (I, \pi_\hb, \com^*, \{\theta^i\}_{i\in I},
 \{\op^{(i,j)}\}_{(i,j)\in I\times[k]})$.

Consider now a verifier $\widetilde{\Ver}_0$ that behaves similarly to
$\Ver$, except that it measures all the registers $\qreg{V^1} \ldots
\qreg{V^\ell}$ in the computational basis. It is easy to see that if $\Pro^*$
convinces $\Ver$ with probability $\mu(\secp)$, then it also convinces
$\widetilde{\Ver}_0$ with probability $\mu(\secp)$. This is because
only the computational basis positions are used for verification by
$\Ver$.

Next, consider a verifier $\tVer$ that behaves similar to $\tVer_0$
except that the registers $\qreg{V^1}, \ldots, \qreg{V^\ell}$ are
measured in the computational basis before $\Pro^*$ is initialized.
Since operations on different registers commute, $\Pro^*$ also
convinces $\tVer$ with probability $\mu(\secp)$:

\[
\Pr\left[
\begin{array}{ll}
\tVer(\crs, \qreg{V}, x, \pi^*) = \top\\
\land\\
x \notin \cL
\end{array}
 \ :
\begin{array}{rl}
 &\crs \gets \Setup(1^\secp)\\
 & (x, \pi^*)\la \Pro^*(\crs, \qreg{P})
\end{array}
\right]
\ge \mu(\secp)
\]
where $\pi^* = (I, \pi_\hb, \com^*, \{\theta^i\}_{i\in I},
 \{\op^{(i,j)}\}_{(i,j)\in I\times[k]})$.

Now, let $y^1, \ldots, y^\ell$ be the values measured by $\tVer$
corresponding to registers $\qreg{V^1}, \ldots, \qreg{V^\ell}$.
Consequently, once the basis choices $\theta^1, \ldots, \theta^\ell$
are specified, the values $t_1, \ldots, t_\ell$ are fixed as they are
computed as the parity of the computational basis values. Observe now
that by statistical binding of $\HBG$, except with negligible
probability, we have that $r_I$ computed by $\tVer$ satisfies
$r_I = t_I \oplus s_I$ where $t_i =
\bigoplus\limits_{j:\ttheta^i_j = 0}y_j^i$ for each $i \in I$ and
$\ttheta^1 \| \ldots \| \ttheta^\ell \seteq \Open(1^{k\ell}, \crs_\bg, \com^*)$. Moreover,
succinctness of $\HBG$ guarantees that $\com^* \in \cC$. Hence, we have
the following, where $\op^{I, [k]}$ denotes $\{\op^{(i,j)}\}_{(i,j)\in
I \times [k]}$:

\begin{align*}
v(\secp)\seteq \Pr\left[
\begin{array}{ll}
\Ver_\hb(I, r_I, x, \pi_\hb) = \top\\
\land \;\; x \notin \cL\\
\land \;\; \com^* \in \cC
\end{array}
 \ :
\begin{array}{rl}
&\crs = (\crs_\bg, s) \gets \Setup(1^\secp)\\
& (x, \pi^*)\la \Pro^*(\crs, \qreg{P})\\
& \ttheta^1 \| \ldots \| \ttheta^\ell \seteq \Open(1^{k\ell},
\crs_\bg, \com^*)\\
& \forall i \in [\ell] : t_i \seteq \bigoplus_{j:\widetilde{\theta}^i_j = 0} y_j^i \\
& r = t \oplus s
\end{array}
\right]
\\
\ge \mu(\secp) - \negl
\end{align*}
where $\pi^* = (I, \pi_\hb, \com^*, \{\theta^i\}_{i\in I},
\{\op^{(i,j)}\}_{(i,j)\in I\times[k]})$.

Next, we define a similar probability for any fixed $\com \in \cC$:

\begin{align*}
v_\com(\secp)&\seteq \Pr\left[
\begin{array}{ll}
\Ver_\hb(I, r_I, x, \pi_\hb) = \top\\
\land \;\; x \notin \cL\\
\land \;\; \com^* = \com
\end{array}
 \ :
\begin{array}{rl}
&\crs = (\crs_\bg, s) \gets \Setup(1^\secp)\\
& (x, \pi^*)\la \Pro^*(\crs, \qreg{P})\\
& \{\ttheta^i\}_{i\in[\ell]} \seteq \Open(1^{k\ell},
\crs_\bg, \com^*)\\
& \forall i \in [\ell] : t_i \seteq \bigoplus_{j:\widetilde{\theta}^i_j = 0} y_j^i \\
& r = t \oplus s
\end{array}
\right]\\
&\le \Pr\left[
\begin{array}{ll}
\Ver_\hb(I, r_I, x, \pi_\hb) = \top\\
\land\\
 x \notin \cL
\end{array}
 \ :
\begin{array}{rl}
&\crs = (\crs_\bg, s) \gets \Setup(1^\secp)\\
& \ttheta^1 \| \ldots \| \ttheta^\ell \seteq \Open(1^{k\ell},
\crs_\bg, \com)\\
& \forall i \in [\ell] : t_i \seteq \bigoplus_{j:\widetilde{\theta}^i_j = 0} y_j^i \\
& r = t \oplus s\\
& (x, I, \pi_\hb) \la \widehat{\Pro}_{\crs_\bg, \com}(r)
\end{array}
\right]
\\
& \le 2^{-q(\secp, n)}\cdot \negl
\end{align*}
where $\pi^* = (I, \pi_\hb, \com^*, \{\theta^i\}_{i\in I},
\{\op^{(i,j)}\}_{(i,j)\in I\times[k]})$.

The first inequality follows because $r$ is already fixed by $\com$
and the second follows from the soundness of the hidden bits NIZK by
viewing $\widehat{\Pro}_{\crs_\bg, \com}(r)$ as a malicious hidden
bits model prover and from the fact that $r$ is uniformly random
(because $s$ is uniform).

Next, we set $q(\secp, n)^{1-\delta} = k \cdot k' \cdot
p(\secp)$ so that:

$$(k\ell)^\delta \cdot p(\secp) = (k\cdot k' \cdot q(\secp, n))^\delta
\cdot p(\secp) \le q^\delta(\secp, n) \cdot q^{1-\delta}(\secp, n)$$

This ensures that $|\cC|\cdot v_\com(\secp) \le \negl$. Therefore,
a union bound gives us the following:

$$\mu(\secp) - \negl \le v(\secp) \le \sum_{\com \in \cC}v_\com(\secp)
\le \negl$$

This proves that $\mu(\secp) \le \negl$, ensuring soundness.

\end{proof}

\normalfont
\paragraph{Acknowledgements:} We would like to thank Fang Song for helpful discussions and the anonymous reviewers of Crypto 2026 and TCC 2026 for their useful feedback.

\paragraph{AI Statement:} The proofs of the EPR construction for computational ZK
(Section \ref{sec:epr-czk}) and CE-ZK (Section \ref{sec:epr-cezk}) have been rewritten using ChatGPT,
fixing the errors from previous versions.

\bibliographystyle{alpha}
\bibliography{bib/abbrev3,references}

%

\end{document}